\RequirePackage{fix-cm}
\documentclass[10pt,letterpaper]{article}
\usepackage{setspace} 
\usepackage{mathptmx}
\usepackage{amssymb}
\usepackage{amsthm}
\usepackage{amsmath}
\usepackage{bm}
\usepackage{graphicx}
\usepackage[table]{xcolor}
\usepackage{multirow}
\usepackage{caption}
\theoremstyle{plain}
\newtheorem{definition}{Definition}[section]

\newtheorem{lemma}{Lemma}[section]

\newtheorem{example}{Example}[section]

\newtheorem{remark}{Remark}[section]
\newtheorem{assumption}{Assumption}[section]
\newcommand{\fe}{\mbox{\rm fe}}
\newcommand{\efe}{\mbox{\rm efe}}
\newcommand{\dis}{\mbox{\rm d}}
\begin{document}
\pagenumbering{arabic}
%
%\hspace{1cm}
%
\begin{center}
{\Large\bf{Quasi-Conscious Multivariate Systems}}\\
Jonathan W.\ D.\ Mason, Mathematical Institute, University of Oxford, UK {\footnotesize(Submitted to Complexity 2015)}
\end{center}
\vspace{-0.7cm}
\begin{spacing}{1}
\begin{abstract}
Conscious experience is awash with underlying relationships. Moreover, for various brain regions such as the visual cortex, the system is biased toward some states. Representing this bias using a probability distribution shows that the system can define expected quantities. The mathematical theory in the present paper links these facts by using expected float entropy (efe), which is a measure of the expected amount of information needed, to specify the state of the system, beyond what is already known about the system from relationships that appear as parameters. Under the requirement that the relationship parameters minimise efe, the brain defines relationships. It is proposed that when a brain state is interpreted in the context of these relationships the brain state acquires meaning in the form of the relational content of the associated experience. For a given set, the theory represents relationships using weighted relations which assign continuous weights, from 0 to 1, to the elements of the Cartesian product of that set. The relationship parameters include weighted relations on the nodes of the system and on their set of states. Examples obtained using Monte-Carlo methods (where relationship parameters are chosen uniformly at random) suggest that efe distributions with long left tails are most important.
\end{abstract}
\end{spacing}
\vspace{-0.7cm}
\section{Introduction}
\label{sec:1}
In the present paper we further develop the theory introduced in the article `Consciousness and the structuring property of typical data' (see \cite{Mason}), and demonstrate and investigate the theory through applications in a number of examples using computational methods.\\
It is intended that the theory will provide a way into the mathematics that underpins how the brain defines the relational content of consciousness. Indeed, conscious experience clings to a substrate of underlying relationships: points in a person's field of view can be strongly related (if close together) or unrelated (if far apart), giving geometry; colours can appear similar (e.g.\ red and orange) or completely different (e.g.\ red and green). We can make a very long list of such examples of relations involving different sounds, smells, tastes and locations of touch. Furthermore, at a higher semantic level involving several brain regions, if we see someone we know and hear a person's name then we know whether the name relates to that person. It is hard to think of any conscious experience that does not involve relations. Whilst it is difficult to explain how the brain defines the colour blue, in the present paper we hope to provide the beginnings of a mathematical theory for how the brain defines all of the relations underlying consciousness and, therefore, explain why, for example, blue appears similar to turquoise but different to red. It is proposed that when a brain state is interpreted in the context of all these relations, defined by the brain, the brain state acquires meaning in the form of the relational content of the experience. If we consider the relations defined by the brain to be a type of statistic then we have the following analogy. A single observation of a one dimensional random variable is almost meaningless, but in the context of the statistics of the random variable, such as mean and variance, the observation has meaning. For arguments in support of this approach, the reader is referred to \cite{Mason}.\\
The issue of how a system such as the brain defines relations is crucial. Importantly, for various brain regions such as the visual cortex, (under temporally well spaced observations of the system) the probability distribution over the different possible states of the system is far from being uniform owing to learning rules of which the Bienenstock, Cooper and Munro (BCM) version of Hebbian theory is one candidate; see \cite{BCM}, \cite{KRB} and \cite{DB}. Hence, the brain is not merely driven by the current sensory input, but is biased toward certain states as a result of a long history of sensory inputs. The probability distribution over the states of the system is therefore a property of the system itself allowing the system to define expected quantities.\\
In the theory presented in the present paper, the brain defines relations under the requirement that the expected quantity of a particular type of entropy is minimised. We call this entropy {\em float entropy}. For a collection of relations on the system and any given state of the system, the float entropy of the state is a measure of the amount of information required, in addition to the information given by the relations, in order to specify that state. We make the definition of float entropy precise in Subsection~\ref{subsec:1.1}. However, later in the present paper we will give a more general definition (multi-relational float entropy) which allows the involvement of more than two relations; see Subsection \ref{subsec:4.1}. We will also consider a time dependent version, and the theory of the present paper will be compared with Integrated Information Theory and Shannon entropy.

\subsection{Definitions}
\label{subsec:1.1}
In this subsection we provide the main definitions in the present paper. Systems such as the brain, and its various regions, are networks of interacting nodes. In the case of the brain we may take the nodes of the system to be the individual neurons or possibly larger structures such as cortical columns. The nodes of the system have a repertoire (range) of states that they can be in. For example, the states that neurons can be in could be associated with different firing frequencies. In the present paper we assume that the node repertoire is finite (as was assumed in \cite{Mason}), and the state of the system is the aggregate of the states of the nodes.\\
The original theory in \cite{Mason} used a mainly set theoretic approach, where a relation on a nonempty set $S$ was usually taken to be a binary relation $R\subseteq S^{2}$. Weighted relations (see below) are slightly more general than binary relations, and the further development (presented in the present paper) of the original theory uses weighted relations because they allow a system to define a weighted relation on the repertoire of its nodes. This is desirable as we will see from examples later in the paper.\\
In Definition \ref{def:S} the elements of the set $S$ are to be taken as the nodes of the system.

\begin{definition}
\label{def:S}
Let $S$ be a nonempty finite set, $n:=\#S$. Then a {\em data element} for $S$ is a set (having a unique arbitrary index label $i$)
\vspace{-0.5cm}
\begin{gather*}
S_{i}:=\{(a,f_{i}(a))\colon a\in S,\, f_{i}:S\to V\},\quad\mbox{where }f_{i}\mbox{ is a map}
\end{gather*}
and $V:=\{v_{1},v_{2},\ldots,v_{m}\}$ is the {\em node repertoire}. The set of all data elements for $S$ given $V$ is $\Omega_{S,V}$ so that $\#\Omega_{S,V}=m^{n}$. 
For temporally well spaced observations, it is assumed that a given finite system defines a random variable with probability distribution $P:\Omega_{S,V}\to [0,1]$ for some finite set $S$ and node repertoire $V$.
If $T$ is a finite set of numbered observations of the system then $T$ is called the {\em typical data} for $S$. The elements of $T$ (called {\em typical data elements}) are handled using a function
\begin{gather*}
\tau:\{1,\ldots,\#T\}\to\{i\colon S_{i}\in\Omega_{S,V}\},
\end{gather*}
where $S_{\tau(k)}$ is the value of observation number $k$ for $k\in\{1,\ldots,\#T\}$. In particular, the function $\tau$ need not be injective since small systems may be in the same state for several observations.
\end{definition}

\begin{remark}
\label{rem:P}
Note that $P$ in Definition \ref{def:S} extends to a probability measure on the power set $2^{\Omega_{S,V}}$ of $\Omega_{S,V}$ by defining
\begin{gather*}
P(A):=\sum_{S_{i}\in A}P(S_{i}),\quad\mbox{for }A\in 2^{\Omega_{S,V}}.
\end{gather*}
Hence, we have a probability space $(\Omega_{S,V},2^{\Omega_{S,V}},P)$ with sample space $\Omega_{S,V}$, sigma-algebra $2^{\Omega_{S,V}}$, and probability measure~$P$.
\end{remark}

We now need the definition of a weighted relation.
\begin{definition}[Weighted relations]
\label{def:WR}
Let $S$ be a nonempty set. A {\em weighted relation} on $S$ is a function of the form
\begin{gather*}
R:S^{2}\to[0,1],
\end{gather*}
where $[0,1]$ is the unit interval. We say that $R$ is:
\begin{enumerate}
\item[1.] {\em reflexive} if $R(a,a)=1$ for all $a\in S$;
\item[2.] {\em symmetric} if $R(a,b)=R(b,a)$ for all $a,b\in S$.
\end{enumerate}
The set of all reflexive, symmetric weighted-relations on $S$ is denoted $\Psi_{S}$.
\end{definition}

\begin{remark}
\label{rem:WR}
Except where stated, the weighted relations used in the present paper are reflexive and symmetric. Relative to such a weighted relation, the value $R(a,b)$ quantifies the strength of the relationship between $a$ and $b$, interpreted in accordance with the usual order structure on $[0,1]$ so that $R(a,b)=1$ is a maximum. For a small finite set, it is useful to display a weighted relation on that set as a weighted relation table (i.e.\ as a matrix).
\end{remark}

Before giving the definition of float entropy we require Definitions \ref{def:U} and \ref{def:d1}.
\begin{definition}
\label{def:U}
Let $S$ be as in Definition \ref{def:S} and let $U:V^{2}\to[0,1]$ be a reflexive, symmetric weighted-relation on the node repertoire $V$; i.e.\ $U\in\Psi_{V}$. Then, for each data element $S_{i}\in\Omega_{S,V}$, we define a function $R\{U,S_{i}\}:S^{2}\to[0,1]$ by setting
\begin{gather*}
R\{U,S_{i}\}(a,b):=U(f_{i}(a),f_{i}(b))\quad\mbox{for all }a,b\in S.
\end{gather*}
It is easy to see that $R\{U,S_{i}\}\in\Psi_{S}$.
\end{definition}

\begin{definition}
\label{def:d1}
Let $S$ be a nonempty finite set. Every weighted relation on $S$ can be viewed as a $\#S^{2}$-dimensional real vector. Hence, the $\dis_{n}$ metric is a metric on the set of all such weighted relations by setting
\begin{gather*}
\dis_{n}(R,R'):=\bigg(\sum_{(a,b)\in S^{2}}|R(a,b)-R'(a,b)|^{n}\bigg)^{1/n},
\end{gather*}
where $R$ and $R'$ are any two weighted relations on $S$. Similarly we have the metric $\dis_{\infty}(R,R'):=\max_{S^{2}}|R(a,b)-R'(a,b)|$.
\end{definition}

\begin{definition}[Float entropy]
\label{def:fe}
Let $S$ be as in Definition \ref{def:S}, let $U\in\Psi_{V}$, and let $R\in\Psi_{S}$. The {\em float entropy} of a data element $S_{i}\in\Omega_{S,V}$, relative to $U$ and $R$, is defined as
\begin{gather*}
\fe(R,U,S_{i}):=\log_{2}(\#\{S_{j}\in\Omega_{S,V}\colon \dis(R,R\{U,S_{j}\})\leq \dis(R,R\{U,S_{i}\})\}),
\end{gather*}
where, in the present paper (unless otherwise stated), $\dis$ is the $\dis_{1}$ metric. Furthermore, let $P:\Omega_{S,V}\to[0,1]$ and $T$ be as in Definition \ref{def:S}. The {\em expected float entropy}, relative to $U$ and $R$, is defined as
\begin{gather*}
\efe(R,U,P):=\sum_{S_{i}\in\Omega_{S,V}}P(S_{i})\fe(R,U,S_{i}).
\end{gather*}
The $\efe(R,U,T)$ approximation of $\efe(R,U,P)$ is defined as
\begin{gather*}
\efe(R,U,T):=\frac{1}{\#T}\sum_{k=1}^{\#T}\fe(R,U,S_{\tau(k)}),
\end{gather*}
where $\tau$ need not be injective by Definition \ref{def:S}. By construction, $\efe$ is measured in bits per data element (bpe).
\end{definition}

It is proposed that a system (such as the brain and its subregions) will define $U$ and $R$ (up to a certain resolution) under the requirement that the $\efe$ is minimised. Hence, for a given system (i.e.\ for a fixed $P$), we attempt to find solutions in $U$ and $R$ to the equation
\begin{gather}
\label{equ:minefe}
\efe(R,U,P)=\min_{R'\in\Psi_{S},\,U'\in\Psi_{V}}\efe(R',U',P).
\end{gather}
In practice we replace $\efe(\cdot,\cdot,P)$ in (\ref{equ:minefe}) with $\efe(\cdot,\cdot,T)$.

\begin{remark}
\label{rem:fe}
In Definition \ref{def:fe} the $\dis_{1}$ metric is used. It turns out that, amongst many metrics, a change in metric has only a small effect on the solutions to (\ref{equ:minefe}). There are also plenty of pathological metrics which, when used, will significantly change the solutions to (\ref{equ:minefe}). In Remark \ref{rem:WR} we mentioned that, for a weighted relation, the value of $R(a,b)$ is interpreted in accordance with the usual order structure on $[0,1]$. We argue that the order structure to be used on $[0,1]$ should be determined by the metric that is being used in Definition \ref{def:fe}. Hence, for a pathological metric, whilst the solutions to (\ref{equ:minefe}) will have changed, their interpretation as weighted relations may be largely unchanged when the order structure used on $[0,1]$ is determined by the metric being used (when this makes sense). In practice, we want to use the usual order structure on $[0,1]$, and this requirement limits which metrics should be used in Definition \ref{def:fe}. We will look at the issue of metrics in some detail in Subsection \ref{subsec:3.3}.
\end{remark}

\begin{remark}
\label{rem:new-vs-old}
The theory presented in the present paper uses the definitions in Subsection \ref{subsec:1.1}. Suppose we restricted these definitions so that the only weighted relation we could use on the node repertoire $V$ was the Kronecker delta, and the only elements of $\Psi_{S}$ we could use were weighted relations taking values in the two point set $\{0,1\}$. Then, under these restrictions, Definition \ref{def:fe} would yield a definition of float entropy equivalent to that given in \cite{Mason}. Indeed, note that a weighted relation $R:S^{2}\to\{0,1\}$ is given by the indicator function for the relation $\{(a,b)\in S^{2}\colon R(a,b)=1\}\subseteq S^{2}$. Hence, the theory presented in the present paper is indeed a development of the theory presented in \cite{Mason}.
\end{remark}

\begin{remark}
\label{rem:cefe}
With reference to Remark \ref{rem:P} and Definition \ref{def:fe}, for $A\in 2^{\Omega_{S,V}}$, we have the {\em weak conditional} $\efe$
\begin{gather*}
\efe(R,U,P\mid A):=\sum_{S_{i}\in\Omega_{S,V}}P(S_{i}\mid A)\fe(R,U,S_{i}).
\end{gather*}
Weak conditional $\efe$ can be useful when considering a system that has entered a particular mode such that this mode restricts the system to a particular set of data elements. There may be other useful definitions of conditional $\efe$.
\end{remark}

\subsection{Advantages of the theory and overview}
\label{subsec:1.2}
The examples in the present paper are intended to have relevance to the visual cortex and our experience of monocular vision. In lieu of typical data for the visual cortex we apply the theory to typical data for digital photographs of the world around us. If the theory, as used in the examples, is relevant to the visual cortex then the examples show that the perceived relationships between different colours, the perceived relationships between different brightnesses, and the perceived relationships between different points in a person's field of view (giving geometry) are all defined by the brain in a mutually dependent way. Hence, in this case, there is a connection between the relationships that underly colour perception and our perception of the underlying geometry of the world around us. Of course the states of the visual cortex are somewhat more complicated than digital photographs since some neurons have sophisticated receptive fields. However, the theory presented in the present paper does not assume that the nodes of the visual cortex have to be individual neurons. Instead, each node can consist of many neurons; effectively representing the data elements using a larger base (note that we can think of the node repertoire as being analogous to a choice of base in the representation of integers). Hence, the examples could well be relevant to the visual cortex. A preliminary discussion and investigation regarding base is presented in Subsection \ref{subsec:3.1}.\\
We also apply the theory to a system where the probability distribution $P$ in Definition \ref{def:S} is uniform over $\Omega_{S,V}$. In this case the solutions to (\ref{equ:minefe}) vary greatly (instead of all being similar) and, hence, the system fails to define weighted relations that give a coherent interpretation of the states of the system. The variation in the solutions to (\ref{equ:minefe}) is partly due to symmetries, and this is discussed in Example \ref{exa:Output_5.6}.\\
It is argued in \cite{Mason} that the theory presented there provides a solution to the binding problem and avoids the homunculus fallacy. Those arguments also apply to the theory presented in the present paper. In particular, consciousness is not the output of some algorithmic process but it may instead, largely, be the states of the system interpreted in the context of the weighted relations that minimise expected float entropy, where here we are talking about a definition of float entropy that involves more than the two weighted relations used in (\ref{equ:minefe}); see Subsection \ref{subsec:4.1}. This argument may become clearer for the reader after going through the examples in the present paper.
The rest of the paper is organised as follows. Section \ref{sec:2} looks at obtaining typical data from digital photographs, and specifies the computational methods used for finding solutions to (\ref{equ:minefe}). Section \ref{sec:3} provides six examples in which the theory is applied. We continue the development of the theory by looking at changing the base of a system, joining and partitioning systems, and metric independence. Section \ref{sec:4} provides generalisations of Definition \ref{def:fe}, a comparison between the present theory and both Giulio Tononi's Integrated Information Theory (IIT) and Shannon entropy, followed by the conclusion. Appendix \ref{App:AppendixA} lists the software used, and Appendix \ref{App:AppendixB} provides a list of notation.
\section{Typical data and computational methods}
\label{sec:2}
In this section we look at obtaining typical data from digital photographs, a binary search algorithm for finding solutions to~(\ref{equ:minefe}), and using $\efe$-histograms to assess guesses when guessing solutions to (\ref{equ:minefe}).
\subsection{Typical data from digital photographs}
\label{subsec:2.1}
When obtaining a typical data element from a digital photograph, in the present paper, only a small part of the photograph is used. This is because the computational methods used in the present paper are suitable for small systems $(\#\Omega_{S,V}\leq10^{6})$ although, at the expense of clarity and ease of implementation, other more efficient computational methods are possible for investigating larger systems; see Appendix \ref{App:AppendixA} which lists the software used during the research for the present paper and provides a discussion on more efficient computational methods.\\
\begin{figure}[b!]
\centering
\includegraphics[width=0.7\textwidth]{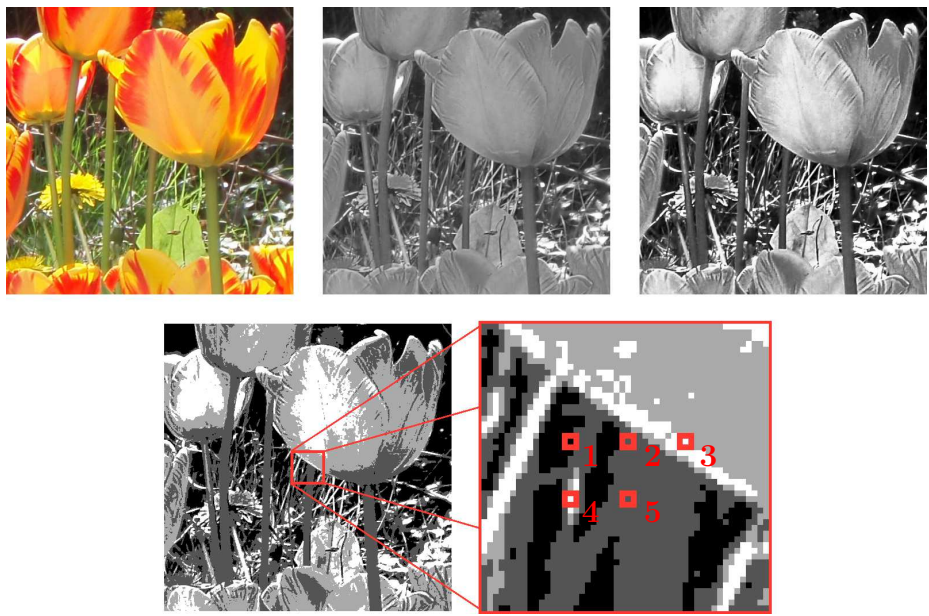}
\caption{Digital photograph sampling using five nodes and a four shade gray scale.}
\label{fig:1}
\end{figure}
\begin{figure}[t!]
\centering
\includegraphics[width=0.7\textwidth]{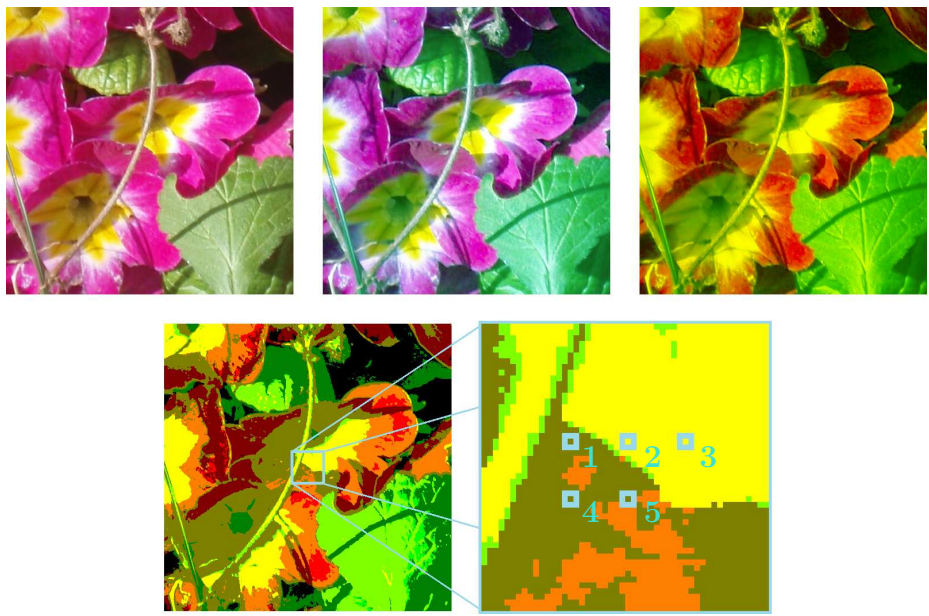}
\caption{Digital photograph sampling using five nodes and a nine colour red/green palette.}
\label{fig:2}
\end{figure}
Figure \ref{fig:1} shows the sampling of a digital photograph such that the typical data element obtained is for a system comprised of five nodes with a four state node repertoire $(\#\Omega_{S,V}=1024)$. Also, in the case of Figure \ref{fig:1}, we are using pixel brightness to determine node state. From top-left to bottom-right, the first image is the original. This image is desaturated (the colours are turned into shades of gray) and then the contrast is enhanced. The contrast enhancement is not required, but it was thought that it might reduce the number of typical data element needed in order to obtain meaningful results. Indeed, when similar, the solutions to (\ref{equ:minefe}) are rather like a type of statistic and, therefore, when using typical data we need to make sure that the sample size is large enough. The image is then posterised (in this case the number of shades is reduced to four giving a four state node repertoire). Finally, five pixels are sampled giving the state of each of the five nodes of the typical data element; see Table \ref{tab:1}. 
\begin{table}[h]
\small
\centering
\caption{Node states of the typical data element obtained from the sampling in Figure \ref{fig:1}.}\smallskip
\begin{tabular}{r|ccccc}
&node 1&node 2&node 3&node 4&node 5\\ \hline
$S_{\tau(1)}$&0.000&147.224&441.673&441.673&147.224
\end{tabular}
\label{tab:1}
\end{table}
To obtain the typical data for the system, this way of obtaining typical data elements is used for several hundred digital photographs. Importantly, what ever the geometric layout of the pixel sampling locations (in Figure~\ref{fig:1} the layout is part of a grid that has adjacent locations every ten pixels), the same layout must be used for all of the digital photographs. Similarly, the same criteria must be used for determining the node states.\\
The sampling in Figure \ref{fig:2} obtains a typical data element for a system comprised of five nodes with a nine state node repertoire $(\#\Omega_{S,V}=59049)$. Here node state is determined by pixel colour over a red/green palette. From top-left to bottom-right, we first have the original image to which colour contrast enhancement is applied. The image is then restricted to colours made up of red and green by setting blue values to zero. The image is then posterised (three values for red and three values for green are used giving a nine state node repertoire). Finally, five pixels are sampled; the result is given in Table \ref{tab:2}.
\begin{table}[ht]
\small
\centering
\caption{Node states of the typical data element obtained from the sampling in Figure \ref{fig:2}.}\smallskip
\begin{tabular}{r|ccccc}
&node 1&node 2&node 3&node 4&node 5\\ \hline
$S_{\tau(1)}$&128,128&255,255&255,255&128,128&128,128
\end{tabular}
\label{tab:2}
\end{table}
We now consider computational methods for finding solutions to (\ref{equ:minefe}).
\subsection{Binary search algorithm}
\label{subsec:2.2}
For any given system, let $n=\#S$ and $m=\#V$.
\begin{enumerate}
\item[Step 1.]
The initial approximation of a solution to (\ref{equ:minefe}) is taken to be the pair $U\in\Psi_{V}$ and $R\in\Psi_{S}$ with $U(v,v')=\tfrac{1}{2}$ and $R(a,b)=\tfrac{1}{2}$ for all $v,v'\in V$, $v\neq v'$, and $a,b\in S$, $a\neq b$, respectively.
\item[Step 2.]
For $U$ and $R$ (shown in Table \ref{tab:3}) a given approximate solution to (\ref{equ:minefe}), let $k=2^{-(q+1)}$ where $q=\min\{i\in\mathbb{N}\colon 2^{i}u_{1,2}\in\mathbb{N}\}$.
\begin{table}[ht]
\small
\centering
\caption{Approximate solution to Equation (\ref{equ:minefe}).}\smallskip
\begin{tabular}{c|cccc}
$U$&$v_{1}$&$v_{2}$&$v_{3}$&$\cdots$\\ \hline
$v_{1}$&1&$u_{1,2}$&$u_{1,3}$&$\cdots$\\
$v_{2}$&$u_{2,1}$&1&$u_{2,3}$&$\cdots$\\
$v_{3}$&$u_{3,1}$&$u_{3,2}$&1&$\cdots$\\
$\vdots$&$\vdots$&$\vdots$&$\vdots$&$\ddots$
\end{tabular}
\qquad
\begin{tabular}{c|cccc}
$R$&node 1&node 2&node 3&$\cdots$\\ \hline
node 1&1&$r_{1,2}$&$r_{1,3}$&$\cdots$\\
node 2&$r_{2,1}$&1&$r_{2,3}$&$\cdots$\\
node 3&$r_{3,1}$&$r_{3,2}$&1&$\cdots$\\
$\vdots$&$\vdots$&$\vdots$&$\vdots$&$\ddots$
\end{tabular}
\label{tab:3}
\end{table}
We now calculate the $\efe$ value of the system for each combination of the entries in Table \ref{tab:4} that give symmetric weighted relations. This is a binary search in the sense that there are two options per entry.
\begin{table}[ht]
\small
\centering
\caption{Binary entries over which to search for approximate solutions to (\ref{equ:minefe}).}\smallskip
\begin{tabular}{c|cccc}
$U$&$v_{1}$&$v_{2}$&$v_{3}$&$\cdots$\\ \hline
$v_{1}$&1&$u_{1,2}\pm k$&$u_{1,3}\pm k$&$\cdots$\\
$v_{2}$&$u_{2,1}\pm k$&1&$u_{2,3}\pm k$&$\cdots$\\
$v_{3}$&$u_{3,1}\pm k$&$u_{3,2}\pm k$&1&$\cdots$\\
$\vdots$&$\vdots$&$\vdots$&$\vdots$&$\ddots$
\end{tabular}
\qquad
\begin{tabular}{c|cccc}
$R$&node 1&node 2&node 3&$\cdots$\\ \hline
node 1&1&$r_{1,2}\pm k$&$r_{1,3}\pm k$&$\cdots$\\
node 2&$r_{2,1}\pm k$&1&$r_{2,3}\pm k$&$\cdots$\\
node 3&$r_{3,1}\pm k$&$r_{3,2}\pm k$&1&$\cdots$\\
$\vdots$&$\vdots$&$\vdots$&$\vdots$&$\ddots$
\end{tabular}
\label{tab:4}
\end{table}
\item[Step 3.]
If the minimum of the $\efe$ values, obtained in Step 2, was given by only one of the pairs of weighted relations tested in Step 2 then redefine $U$ and $R$ as this new pair of weighted relations and return to Step 2. Otherwise, output $U$, $R$ and their associated $\efe$ value, and stop.
\end{enumerate}
If the algorithm did not stop then the chronology of approximate solutions, given by the applications of Step 3, would be a convergent sequence with respect to $d_{1}$ and any of the metrics in Definition \ref{def:d1}. However, for $m\geq2$ and $n\geq2$, both $\Psi_{V}$ and $\Psi_{S}$ are uncountable infinite sets; whereas the number of possible $\efe$ values is finite. Hence, some $\efe$ values result from infinitely many weighted relations in $\Psi_{V}$ and $\Psi_{S}$. It is not surprising then that, as the approximate solutions become closer with respect to $d_{1}$, ultimately the algorithm stops at Step 3. In short, the system defines $U$ and $R$ (up to a certain resolution) under the requirement that the $\efe$ is minimised.\\
This search algorithm works well; see its use in Section \ref{sec:3}. However, the number of $\efe$ values calculated during each application of Step 2 is $2^{(n(n-1)+m(m-1))/2}$. For example, a system with $\#S=\#V=5$ can result in the algorithm calculating more than $10^7$ $\efe$ values before stopping. Hence, the present paper also uses the following, computationally less expensive, method for approximating solutions to (\ref{equ:minefe}); also see Appendix \ref{App:AppendixA} concerning more efficient computational methods.
\subsection{Using $\efe$-histograms obtained from Monte-Carlo methods}
\label{subsec:2.3}
Here we choose $U\in\Psi_{V}$ and $R\in\Psi_{S}$ uniformly at random. With reference to Table \ref{tab:3}, this is done by choosing each off-diagonal upper-triangular entry of $U$ and $R$ uniformly at random from the interval $[0,1]$ (the off-diagonal lower-triangular entries are then those making $U$ and $R$ symmetric). The $\efe$ value is then calculated and stored, and the whole process is repeated producing a list of many thousands of $\efe$ observations from which an $\efe$-histogram can be obtained. With this setup, if we wish to treat $\efe$ as a random variable then standard methods can be used for approximating the probability distribution from the $\efe$ values (although this can be difficult for distributions with very thin tails). In any case, provided enough observations are made, the $\efe$-histogram can be used to help assess guesses when guessing approximate solutions to~(\ref{equ:minefe}).\\
However, we need to be careful concerning what is meant by `choose uniformly at random from the interval $[0,1]$'. Usually, this means that all subintervals of the same length are equally probable events. This is fine for us as long as the length of subintervals is determined by the metric used in Definition \ref{def:fe}, which conveniently is $\dis_{1}$; see Subsection \ref{subsec:3.3} for relevant details.\\
We are now ready to apply the theory.
\section{Examples and investigations}
\label{sec:3}
This section provides insight concerning how the theory performs in practice by way of several informative examples and investigations.
\begin{example}
\label{exa:Output_5.1}
In this example 200 digital photographs of the world around us are used. The typical data is obtained using exactly the method shown in Figure \ref{fig:1}, where the photographs have a four shade gray scale. Hence, $\#T=200$ and the system is comprised of five nodes with a four state node repertoire $(\#\Omega_{S,V}=1024)$. The binary search algorithm of Subsection \ref{subsec:2.2} was applied to $T$ and, after ten cycles, returned the weighted relations in Table \ref{tab:5}. Figure \ref{fig:3} provides a graph illustration of the weighted relations. For $U$, values above 0.2 are indicated with a solid line, whilst values from 0.02 to 0.2 are indicated with a dash line. For $R$, values above 0.9 are indicated with a solid line, whilst values from 0.75 to 0.9 are indicated with a dash line. Although $\#T=200$ is rather small, $T$ has defined the correct relationships under the requirement that $\efe$ is minimised. As described in Subsection \ref{subsec:2.3}, Figure \ref{fig:4} provides an $\efe$-histogram for $T$. For $U$ and $R$ in Table \ref{tab:5} we have $\efe(R,U,T)=4.91623$, to six sf, and this value is indicated in Figure \ref{fig:4} by the triangular marker furthest to the left.
\begin{table}[ht]
\small
\centering
\caption{Approximate solution for Example \ref{exa:Output_5.1}.}\smallskip
\begin{tabular}{c|cccc}
    $U$&            0&      147.224&      294.449&      441.673\\ \hline
      0&            1&0.30908203125&0.05224609375&0.00439453125\\
147.224&0.30908203125&            1&0.41064453125&0.10400390625\\
294.449&0.05224609375&0.41064453125&            1&0.34228515625\\
441.673&0.00439453125&0.10400390625&0.34228515625&1
\end{tabular}\\
\vspace{3mm}
\begin{tabular}{c|ccccc}
   $R$&       node 1&       node 2&       node 3&       node 4&       node 5\\ \hline
node 1&            1&0.99853515625&0.62353515625&0.92041015625&0.78369140625\\
node 2&0.99853515625&            1&0.94580078125&0.75244140625&0.93505859375\\
node 3&0.62353515625&0.94580078125&            1&0.73486328125&0.88330078125\\
node 4&0.92041015625&0.75244140625&0.73486328125&            1&0.98193359375\\
node 5&0.78369140625&0.93505859375&0.88330078125&0.98193359375&1
\end{tabular}
\label{tab:5}
\end{table}
\begin{figure}[ht]
\centering
\includegraphics[width=0.5\textwidth]{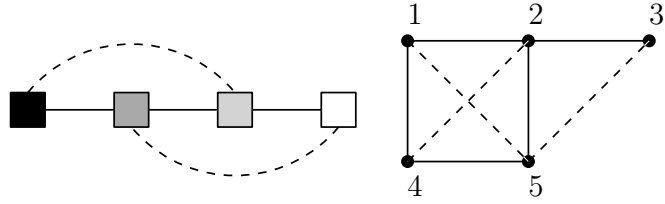}
\captionsetup{justification=centering,margin=2cm}
\caption{Graph illustration of the weighted relations in Table \ref{tab:5}, showing strongest relationships ({\em solid lines}) and intermediate relationships ({\em dash lines}).}
\label{fig:3}
\end{figure}
\begin{figure}[!ht]
\centering
\includegraphics[width=0.75\textwidth]{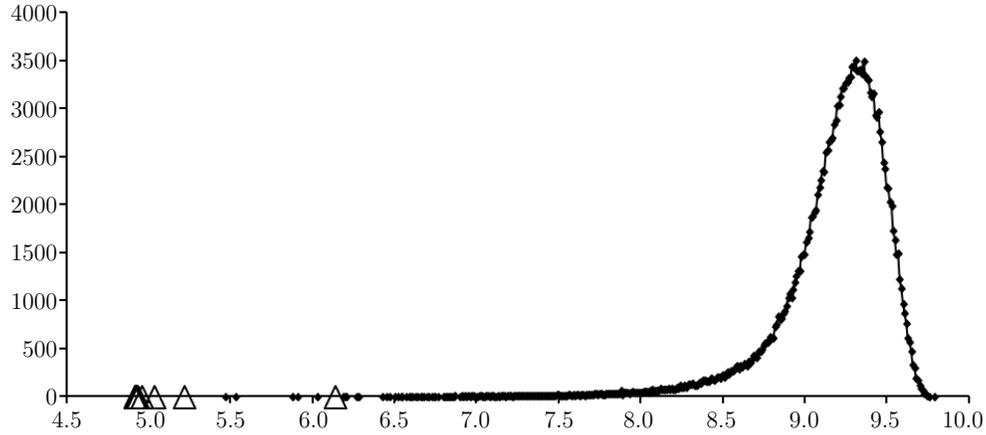}
\captionsetup{justification=centering,margin=2cm}
\caption{An $\efe$-histogram for Example \ref{exa:Output_5.1} using 200,000 observations and a bin interval of 0.01. For each cycle of the binary search algorithm, the $\efe$ value of the approximate solution obtained is shown ({\em triangular marker}).}
\label{fig:4}
\end{figure}
The $\efe$-histogram is negatively skewed with a long left tail and this shape is usual for systems where the probability distribution~$P$, in Definition \ref{def:S}, is far from uniform over $\Omega_{S,V}$.
\end{example}
Example \ref{exa:Output_20} involves a larger system than that of Example \ref{exa:Output_5.1}. Here enlarging the system results in an increase in the difference between the minimum $\efe$ and the location (mean or median) of the $\efe$-histogram.
\begin{example}
\label{exa:Output_20}
In this example 400 digital photographs of the world around us are used. The typical data is obtained using the method shown in Figure \ref{fig:1}, except the number of sampling locations is increased from five to nine to form a three by three grid. Since, $\#T=400$ and $\#\Omega_{S,V}=4^{9}=262144$, this system is too large to apply the binary search algorithm. Instead, Table~\ref{tab:5} in Example \ref{exa:Output_5.1} was used to guess an approximate solution. Figure \ref{fig:5} provides an $\efe$-histogram for $T$.
\begin{figure}[!th]
\centering
\includegraphics[width=0.75\textwidth]{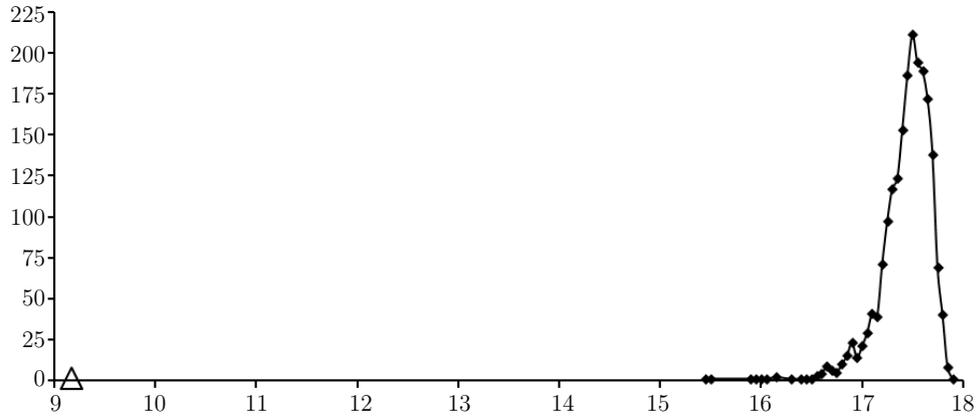}
\captionsetup{justification=centering,margin=2cm}
\caption{An $\efe$-histogram for Example \ref{exa:Output_20} using 2000 observations and a bin interval of 0.05. The $\efe$ value of the approximate solution is shown ({\em triangular marker}).}
\label{fig:5}
\end{figure}
The $\efe$ value for the approximate solution is indicated with a triangular marker and shows that the guess is favorable.
\end{example}
In the next two examples the theory is applied to systems where the probability distribution $P$, in Definition \ref{def:S}, is uniform over $\Omega_{S,V}$. These two examples can be compared with Example \ref{exa:Output_5.1}.
\begin{example}
\label{exa:Output_5.4}
In this example $\#T=200$ and the system is comprised of five nodes with a four state node repertoire $(\#\Omega_{S,V}=1024)$, as is the case in Example \ref{exa:Output_5.1}. However, in the present example, the elements of $T$ are chosen uniformly at random from $\Omega_{S,V}$. Figure \ref{fig:6} provides an $\efe$-histogram for $T$. The binary search algorithm was also applied to $T$ and completed thirteen cycles. The $\efe$-histogram is not negatively skewed and the difference between the $\efe$ value of the approximate solution, found by the binary search algorithm, and the mean of the $\efe$-histogram is only 0.62, to three sf, compared to 4.26 for Example~\ref{exa:Output_5.1}. A second choice for the elements of $T$ was then made uniformly at random from $\Omega_{S,V}$. The approximate solution, found by the binary search algorithm, for the second choice of $T$ was very different to that of the first choice of $T$.
\begin{figure}[!hb]
\centering
\includegraphics[width=0.75\textwidth]{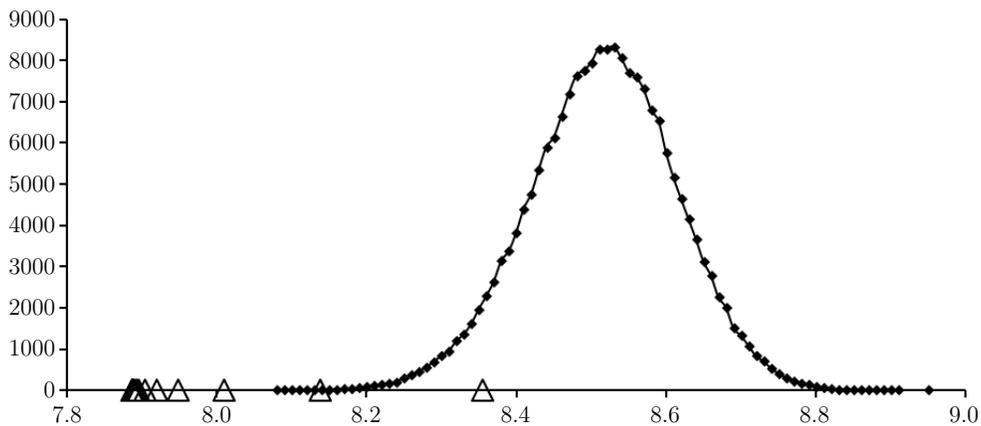}
\captionsetup{justification=centering,margin=2cm}
\caption{An $\efe$-histogram for Example \ref{exa:Output_5.4} using 200,000 observations and a bin interval of 0.01. For each cycle of the binary search algorithm, the $\efe$ value of the approximate solution obtained is shown ({\em triangular marker}).}
\label{fig:6}
\end{figure}
\end{example}
\begin{example}
\label{exa:Output_5.6}
In this example the system is again comprised of five nodes with a four state node repertoire. However, $\#T=1024$ such that there is exactly one observation of each element of $\Omega_{S,V}$ in $T$. In this case if we take the probability distribution $P$, in Definition \ref{def:S}, to be uniform over $\Omega_{S,V}$ then $\efe(\cdot,\cdot,T)=\efe(\cdot,\cdot,P)$; see Definition \ref{def:fe}. In particular, if we let $T'$ denote the typical data in Example \ref{exa:Output_5.4} then the present example is the limit case for Example \ref{exa:Output_5.4} as $\#T'\to\infty$. Figure~\ref{fig:7} provides an $\efe$-histogram for $T$. The binary search algorithm was applied to $T$ but stopped before completing one cycle because the minimum of the $\efe$ values, obtained in Step 2 of the algorithm, was given by many of the pairs of weighted relations tested in Step 2. This is due to a type of symmetry within $T$ which we now consider.\\
We can represent $T$ in the form of a table with each row corresponding to a typical data element; e.g. see Tables \ref{tab:1} and \ref{tab:2}. A transformation of $T$ can be made by, for example, switching the content of columns 3 and 4 of $T$, which is equivalent to switching round the node labels at the top of the columns. Table \ref{tab:6} presents one of the pairs of weighted relations that gave the minimum $\efe$ value obtained in Step 2 of the algorithm. A transformation of $R$ in Table \ref{tab:6} can be made by switching the content of columns 3 and 4, and then switching the content of rows 3 and 4. Clearly, the $\efe$ is invariant under performing both the transformation to $T$ and the transformation to $R$. Now, because $T$ is comprised of exactly one observation of each element of $\Omega_{S,V}$, the rows of the transformed version of $T$ can be reordered to give back $T$ before the transformation was made. Since $\efe$ is invariant regarding the order of typical data elements, the $\efe$ value given by $T$ relative to $U$ and the transformed version of $R$ is the same as $\efe(R,U,T)$. Since $R$ and its transformed version are different, the minimum of the $\efe$ values, obtained in Step 2 of the algorithm, is given by more than one of the pairs of weighted relations tested in Step 2. The same argument also applies to the solutions to (\ref{equ:minefe}) and, consequently, these solutions vary greatly with respect to $\dis_{1}$.\\
Also in the present example, for every fixed $U$ and $R$, the transformation on $T$ is a type of $\efe$ preserving involution (i.e. $T$ has a type of symmetry).
\begin{figure}[!b]
\centering
\includegraphics[width=0.75\textwidth]{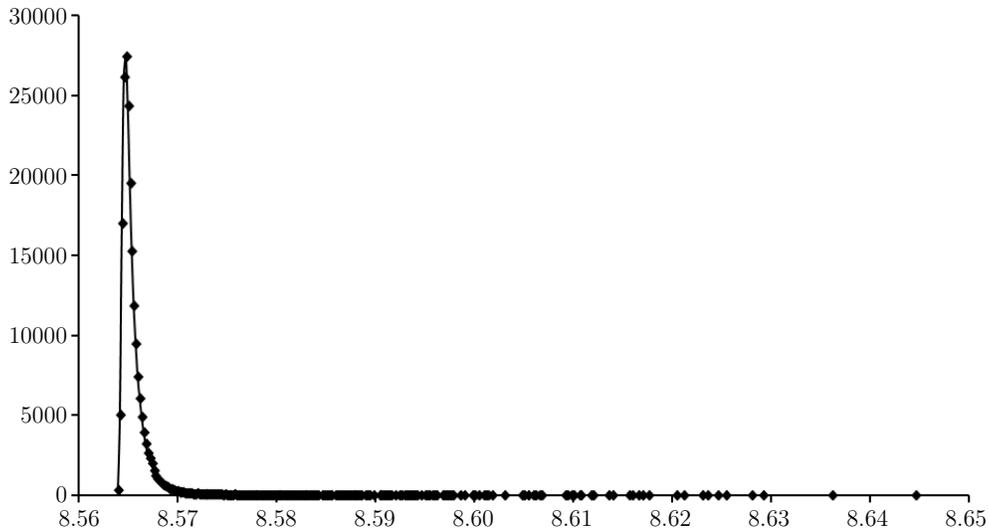}
\captionsetup{justification=centering,margin=2cm}
\caption{An $\efe$-histogram for Example \ref{exa:Output_5.6} using 200,000 observations and a bin interval of 0.0002.}
\label{fig:7}
\end{figure}
\begin{table}[!t]
\small
\centering
\captionsetup{justification=centering,margin=2cm}
\caption{One of the pairs of weighted relations, in Example \ref{exa:Output_5.6}, that gave the minimum $\efe$ value obtained in Step 2 of the binary search algorithm.}\smallskip
\begin{tabular}{c|cccc}
    $U$&$v_{1}$&$v_{2}$&$v_{3}$&$v_{4}$\\ \hline
$v_{1}$&      1&   0.75&   0.75&   0.25\\
$v_{2}$&   0.75&      1&   0.25&   0.25\\
$v_{3}$&   0.75&   0.25&      1&   0.25\\
$v_{4}$&   0.25&   0.25&   0.25&      1
\end{tabular}
\qquad
\begin{tabular}{c|ccccc}
   $R$&node 1&node 2&node 3&node 4&node 5\\ \hline
node 1&     1&  0.25&  0.25&  0.75&  0.25\\
node 2&  0.25&     1&  0.75&  0.25&  0.25\\
node 3&  0.25&  0.75&     1&  0.25&  0.25\\
node 4&  0.75&  0.25&  0.25&     1&  0.25\\
node 5&  0.25&  0.25&  0.25&  0.25&     1
\end{tabular}
\label{tab:6}
\end{table}
More generally beyond the present example, the extent to which $T$ is invariant, up to the order of its rows following such transformations, may be important regarding the shape of the $\efe$-histogram.
\end{example}
\begin{remark}
\label{rem:group}
Note that, in Example \ref{exa:Output_5.6}, the involution on $T$ can be put into a broader context as an element of a group of permutations of the contents of the columns of $T$. Similarly, the transformation applied to $R$ is an element of a group of such transformations on $\Psi_{S}$ . There is also a similar group of transformations on $\Psi_{V}$. Beyond Example \ref{exa:Output_5.6}, for a given system it may be that such a transformation on $\Psi_{S}$ acts almost as the identity on the solutions to (\ref{equ:minefe}). In this case the system has defined geometry on $S$, under the requirement that the $\efe$ is minimised, that has a symmetry such as a rotation or reflection etc.\\
Upon consideration of the positively skewed $\efe$-histogram in Figure \ref{fig:7}, the reader might ask why we do not look for pairs of weighted relations that maximise $\efe$ instead of minimise it. For every given system, the weighted relations $U\in\Psi_{V}$ and $R\in\Psi_{S}$ that maximise $\efe$ are the constant functions which everywhere take the value 1; see Definition \ref{def:fe}.
\end{remark}
In the next example the typical data is obtained from colour digital photographs.
\begin{example}
\label{exa:Output_6.1}
In this example 600 digital photographs of the world around us are used. The typical data is obtained using exactly the method shown in Figure \ref{fig:2}, where the photographs have a nine colour red/green palette. Hence, $\#T=600$ and the system is comprised of five nodes with a nine state node repertoire $(\#\Omega_{S,V}=59049)$. The system is too large to apply the binary search algorithm. Hence, in this case, approximate solutions to (\ref{equ:minefe}) are guessed and their associated $\efe$ values are compared with an $\efe$-histogram for the system. Table \ref{tab:7} presents the guess for $R$; the right hand side of Figure \ref{fig:3} provides a graph illustration for $R$. Table \ref{tab:8} gives two different guesses, $U'$ and $U$, for the weighted relation on $V$ (note that the node repertoire labels are of the form red,green i.e.\ $255,0$ is the label for pure red).
\begin{table}[!ht]
\small
\centering
\captionsetup{justification=centering,margin=2cm}
\caption{Guess for $R$ in Example \ref{exa:Output_6.1}.}\smallskip
\begin{tabular}{c|ccccc}
   $R$&node 1&node 2&node 3&node 4&node 5\\ \hline
node 1&     1&  0.95&  0.65&  0.95&  0.75\\
node 2&  0.95&     1&  0.95&  0.75&  0.95\\
node 3&  0.65&  0.95&     1&  0.60&  0.75\\
node 4&  0.95&  0.75&  0.60&     1&  0.95\\
node 5&  0.75&  0.95&  0.75&  0.95&     1
\end{tabular}
\label{tab:7}
\end{table}
\begin{table}[!tp]
\small
\centering
\captionsetup{justification=centering,margin=2cm}
\caption{Guesses for the weighted relation on $V$ in Example \ref{exa:Output_6.1}.}\smallskip
\begin{tabular}{r|ccccccccc}
   $U'$& 0,0&128,0&255,0&0,128&128,128&255,128&0,255&128,255&255,255\\ \hline
    0,0&   1& 0.45& 0.15& 0.45&   0.25&   0.10& 0.15&   0.10&   0.05\\
  128,0&0.45&    1& 0.45& 0.25&   0.45&   0.25& 0.10&   0.15&   0.10\\
  255,0&0.15& 0.45&    1& 0.10&   0.25&   0.45& 0.05&   0.10&   0.15\\
  0,128&0.45& 0.25& 0.10&    1&   0.45&   0.15& 0.45&   0.25&   0.10\\
128,128&0.25& 0.45& 0.25& 0.45&      1&   0.45& 0.25&   0.45&   0.25\\
255,128&0.10& 0.25& 0.45& 0.15&   0.45&      1& 0.10&   0.25&   0.45\\
  0,255&0.15& 0.10& 0.05& 0.45&   0.25&   0.10&    1&   0.45&   0.15\\
128,255&0.10& 0.15& 0.10& 0.25&   0.45&   0.25& 0.45&      1&   0.45\\
255,255&0.05& 0.10& 0.15& 0.10&   0.25&   0.45& 0.15&   0.45&      1
\end{tabular}\\
\vspace{3mm}
\begin{tabular}{r|ccccccccc}
    $U$& 0,0&128,0&255,0&0,128&128,128&255,128&0,255&128,255&255,255\\ \hline
    0,0&   1& 0.11& 0.04& 0.11&   0.45&   0.09& 0.04&   0.09&   0.16\\
  128,0&0.11&    1& 0.11& 0.04&   0.11&   0.45& 0.02&   0.04&   0.09\\
  255,0&0.04& 0.11&    1& 0.02&   0.04&   0.11&0.015&   0.02&   0.04\\
  0,128&0.11& 0.04& 0.02&    1&   0.11&   0.04& 0.11&   0.45&   0.09\\
128,128&0.45& 0.11& 0.04& 0.11&      1&   0.11& 0.04&   0.11&   0.45\\
255,128&0.09& 0.45& 0.11& 0.04&   0.11&      1& 0.02&   0.04&   0.11\\
  0,255&0.04& 0.02&0.015& 0.11&   0.04&   0.02&    1&   0.11&   0.04\\
128,255&0.09& 0.04& 0.02& 0.45&   0.11&   0.04& 0.11&      1&   0.11\\
255,255&0.16& 0.09& 0.04& 0.09&   0.45&   0.11& 0.04&   0.11&      1
\end{tabular}
\label{tab:8}
\end{table}
Figure \ref{fig:8} provides a graph illustration of the weighted relations in Table \ref{tab:8}. Figure \ref{fig:9} provides an $\efe$-histogram for $T$ and, whilst $U'$ is an obvious first guess, it is $U$ that gives the lower $\efe$ value.
\begin{figure}[!hp]
\centering
\includegraphics[width=0.45\textwidth]{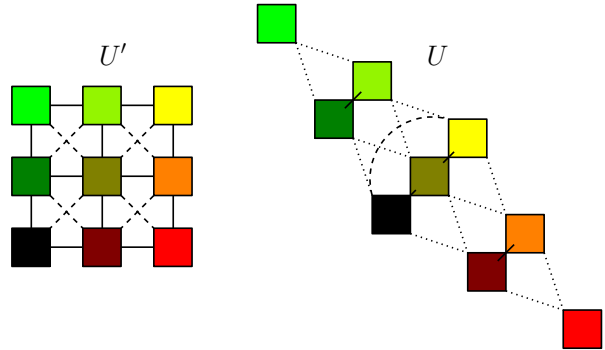}
\captionsetup{justification=centering,margin=1cm}
\caption{Graph illustration of the weighted relations in Table \ref{tab:8}, showing strongest relationships ({\em solid lines}), intermediate relationships ({\em dash lines}) and, for $U$ only, weak intermediate relationships ({\em dotted lines}).}
\label{fig:8}
\end{figure}
\begin{figure}[!hp]
\centering
\includegraphics[width=0.75\textwidth]{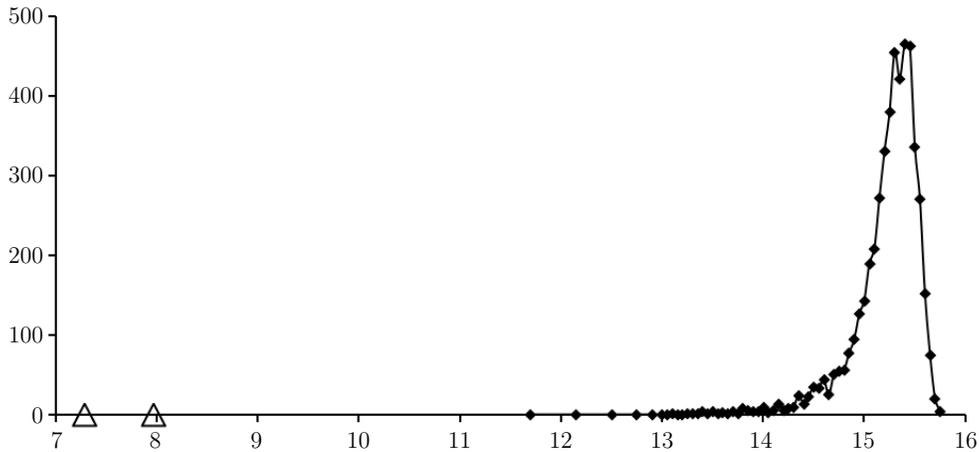}
\captionsetup{justification=centering,margin=1cm}
\caption{An $\efe$-histogram for Example \ref{exa:Output_6.1} using 5000 observations and a bin interval of 0.05. The values $\efe(R,U',T)=7.97695$ and $\efe(R,U,T)=7.28947$, to six sf, are shown ({\em triangular marker}).}
\label{fig:9}
\end{figure}
With respect to $U$, elements of $V$ of the form $x,x+a$, where $a$ is constant, are more strongly related than elements of the form $x,a-x$; i.e. with respect to $U$, the representative of pure red is very distinct from the representative of pure green. We note that, given the $\efe$-histogram and the system, $R$ and $U$ appear to be favorable and appropriate weighted relations. However, $R$ and $U$ are still only guesses and actual solutions to (\ref{equ:minefe}) could be somewhat different. Ideally we would use the binary search algorithm on a similar system but with a larger node repertoire and a larger typical data, but this comes at a high computational expense.
\end{example}
We now have the first of three investigations concerning the theory.
\subsection{Changing the base of a system}
\label{subsec:3.1}
In this subsection we look at base changing operations, base branching structure, and the affect of changing base on $\efe$-histograms.
\subsubsection{Base changing operations}
\label{subsubsec:BCO}
Here we look at two different types of base changing operations. One of the types of operations involves combining nodes whilst the other involves splitting nodes. Many operations of the same type are equivalent in the sense that the resulting systems only differ in the choice of labels used for nodes or repertoire elements. Furthermore, every combining operation is the inverse of some splitting operation and vice versa. As an example, suppose we have a system with $\#S=6$ and $\#V=2$. Table \ref{tab:9} shows one of the ways to combine the nodes so that the resulting system has $\#S'=3$ and $\#V'=4$.
\begin{table}[h]
\small
\centering
\captionsetup{justification=centering,margin=2cm}
\caption{Example of changing the base of a system.}\smallskip
\begin{tabular}{c|cccccc}
             &  node 1&  node 2&  node 3&  node 4&  node 5& node 6\\ \hline
$S_{\tau(1)}$& $v_{2}$& $v_{2}$& $v_{1}$& $v_{1}$& $v_{2}$&$v_{1}$\\
     $\vdots$&$\vdots$&$\vdots$&$\vdots$&$\vdots$&$\vdots$&$\vdots$
\end{tabular}\\
\vspace{3mm}
\begin{tabular}{c|ccc}
               &         node 1&         node 2&          node 3\\
Node allocation&  (node 1,node 4)&  (node 5,node 2)&  (node 6,node 3)\\ \hline
 $S'_{\tau(1)}$&($v_{2}$,$v_{1}$)&($v_{2}$,$v_{2}$)&($v_{1}$,$v_{1}$)\\
       $\vdots$&       $\vdots$&       $\vdots$&         $\vdots$
\end{tabular}\\
\vspace{3mm}
\begin{tabular}{ccccc}
                     &         $v'_{1}$&         $v'_{2}$&         $v'_{3}$&     $v'_{4}$\\
Repertoire allocation&($v_{2}$,$v_{2}$)&($v_{2}$,$v_{1}$)&($v_{1}$,$v_{1}$)&($v_{1}$,$v_{2}$)
\end{tabular}\\
\vspace{3mm}
\begin{tabular}{c|ccc}
               &         node 1&         node 2&         node 3\\ \hline
 $S'_{\tau(1)}$&       $v'_{2}$&       $v'_{1}$&       $v'_{3}$\\
       $\vdots$&       $\vdots$&       $\vdots$&         $\vdots$
\end{tabular}
\label{tab:9}
\end{table}
More generally, from Table \ref{tab:9}, we see that there are $6!$ different possible node allocations and~$4!$ different possible repertoire allocations. Hence, in this case, the total number of such combining operations (or splitting operations if reversing the process) is $6!4!=17280$, and this corresponds to the fact that $\#\Omega_{S,V}=2^{6}=4^{3}=\#\Omega_{S',V'}$. Similarly we note that $\#\Omega_{S,V}=2^{6}=8^{2}$ so that there are $6!8!$ different combining operations resulting in systems with two nodes and an eight state node repertoire.\\
Such operations do have an affect on $\efe$-histograms. Indeed, since $\#\Omega_{S,V}=2^{6}=64^{1}$, we can apply a combining operation that results in a system with one node and a node repertoire that has a state for every state of the system. The resulting $\efe$-histogram has a standard deviation of zero and is located at the maximum possible $\efe$ value for a system with 64 states, which is $\log_{2}(64)=6$; see Definition \ref{def:fe} and Subsection \ref{subsec:2.3}. We will further consider the affect of base changing operations on $\efe$-histograms in Subsection \ref{subsubsec:BEH}.
\subsubsection{Base branching structure}
\label{subsubsec:BBS}
We have already noted in Subsection \ref{subsubsec:BCO} that many base changing operations are equivalent in the sense that the resulting systems only differ in the choice of labels used for nodes or repertoire elements. The advantage of this redundancy, for appropriate systems, is that it allows us to apply a splitting operation in the first instance (i.e. we can start with a repertoire allocation) instead of being restricted to combining operations. Alternatively we can avoid this redundancy by treating a system in its initial base as being at the bottom of a branching structure which branches under combining nodes such that each branch terminates with the system represented by a single node. Table \ref{tab:10} shows one such branch. We note that, with regard to weighted relations on the nodes, the order of the columns in Table \ref{tab:10} is not important as long as column heading and column contents are kept together. Furthermore, there is no repertoire allocation since we retain the vector form of the node states. These simplifications reduce the number of combining operations discussed in Subsection \ref{subsubsec:BCO} from 17280 to~120.
\begin{table}[h]
\small
\centering
\captionsetup{justification=centering,margin=2cm}
\caption{One branch of a base branching structure.}\smallskip
\begin{tabular}{c|cccccc}
             &  node 1&  node 2&  node 3&  node 4&  node 5& node 6\\ \hline
$S_{\tau(1)}$& $v_{2}$& $v_{2}$& $v_{1}$& $v_{1}$& $v_{2}$&$v_{1}$\\
     $\vdots$&$\vdots$&$\vdots$&$\vdots$&$\vdots$&$\vdots$&$\vdots$
\end{tabular}\\
\vspace{3mm}
\begin{tabular}{c|ccc}
        Branch&  (node 1,node 4)&  (node 5,node 2)&  (node 6,node 3)\\ \hline
$S'_{\tau(1)}$&($v_{2}$,$v_{1}$)&($v_{2}$,$v_{2}$)&($v_{1}$,$v_{1}$)\\
      $\vdots$&       $\vdots$&       $\vdots$&         $\vdots$
\end{tabular}\\
\vspace{3mm}
\begin{tabular}{c|c}
  End of Branch&((node 5,node 2),(node 1,node 4),(node 6,node 3))\\ \hline
$S''_{\tau(1)}$&(($v_{2}$,$v_{2}$),($v_{2}$,$v_{1}$),($v_{1}$,$v_{1}$))\\
       $\vdots$&$\vdots$
\end{tabular}
\label{tab:10}
\end{table}\\
Now the definition of float entropy in Definition \ref{def:fe} uses only one base for a system. However, multi-relational float entropy (see Subsection \ref{subsec:4.1}) involves more weighted relations by involving more than one base. For some systems it may be that particular bases are important regarding weighted relations that minimise $\efe$ and/or regarding maximising the length of the left tail of the $\efe$-histogram. Indeed, we have already noted that combining all of the nodes of a system into a single node is not good in this respect, showing that other bases are preferable; see Subsection \ref{subsubsec:BCO}. Moreover, a change of base may allow a system to define weighted relations at a higher level of meaning. For example, the solutions to (\ref{equ:minefe}) may define (to a high resolution) a weighted relation $R$ on the nodes of some system, giving two dimensional geometry. For a particular branch of the base branching structure, the states of the composite nodes will be images under the geometry (given by $R$) on the nodes that have been combined. Hence, under the requirement of further minimising $\efe$, the system defines a weighted relation on the repertoire of the composite nodes and hence on a set of images; see Subsection \ref{subsec:4.1}. This may have relevance to some aspects of the Gestalt theory of visual perception; see \cite{Wagemans}.\\
Comparing base changing operations with base branching structure we note that allowing arbitrary application of successive combining and splitting operations may provide too much freedom in the sense that a system may then define too many weighted relations (under requirements such as the minimisation of $\efe$) to specify a single consistent interpretation of the states of the system. Hence, restricting the theory to the base branching structure may be desirable (or perhaps at least to some further generalisation of the base branching structure). In spite of this we will now look at the affect of combining nodes and splitting nodes on $\efe$-histograms.
\subsubsection{Base and $\efe$-histograms}
\label{subsubsec:BEH}
The following lemma says that uniform randomness is preserved by both combining and splitting operations.
\begin{lemma}
\label{lem:randtorand}
Suppose we have a system where the probability distribution $P$ in Definition \ref{def:S} is uniform over $\Omega_{S,V}$. For any given combining or splitting operation, as described in Subsection \ref{subsubsec:BCO}, let $S'$ and $V'$ (with $\#V'$ as small as possible) be such that $\Omega_{S',V'}$ is the codomain of $\Omega_{S,V}$ under the operation. Furthermore, for $S'_{i}$ an element of the image of $\Omega_{S,V}$ under the operation, define $P'(S'_{i}):=P(A_{S'_{i}})$, where $A_{S'_{i}}$ is the preimage of $S'_{i}$. Then $P'$ is a uniform probability distribution over $\Omega_{S',V'}$.
\end{lemma}
\begin{proof}
Immediate since the operation is a bijection from $\Omega_{S,V}$ to $\Omega_{S',V'}$.
\end{proof}
We now consider the case where $P$ is far from uniform over $\Omega_{S,V}$. Because computational recourses are limited, a choice had to be made between looking at the affects of many different base changing operations on just one system and looking at the affects of one or two different base changing operations per system for several different systems. The latter was chosen in order to avoid accidentally giving results for some highly unusual system. Typical data was obtained for each of the systems from digital photographs. The method in Figure \ref{fig:1} was used except the number of shades in the gray scale, the location of the sampling grid and the number of nodes involved varied from system to system; details are given on the left-hand side of Table \ref{tab:11}.
\begin{table}[h]
\small
\centering
\captionsetup{justification=centering,margin=2cm}
\caption{Seven systems from which $\efe$ observations were taken both before and after the application of a base changing operation.}\smallskip
\begin{tabular}{cccr|c|ccr}
System&$\#S$&$\#V$&$\#\efe$-observations&Operation&$\#S'$&$\#V'$&$\#\efe$-observations\\ \hline
     1&    6&    2&               400000&  combine&     3&     4&               400000\\
     2&    3&    4&                 5000&    split&     6&     2&                 5000\\
     3&    3&    4&                 5000&    split&     6&     2&                 5000\\
     4&    3&    9&               400000&    split&     6&     3&               400000\\
     5&    4&    9&               100000&    split&     8&     3&               100000\\
     6&    6&    3&               400000&  combine&     3&     9&               400000\\
     7&    6&    3&               400000&  combine&     3&     9&               400000
\end{tabular}
\label{tab:11}
\end{table}\\
For each of the systems, 400 typical data elements were collected. Subsequently a large number of $\efe$ observations were obtained using the method described in Subsection \ref{subsec:2.3}. The same number of $\efe$ observations was then obtained having applied a base changing operation to the typical data. Apart from the size of base (i.e.\ the size of the repertoire $\#V'$ in Table~\ref{tab:11}), the base changing operation was chosen at random for each system. With respect to the seven systems in Table \ref{tab:11}, we note that System 3 is actually the same as System 2 in the sense that the same typical data is used. However, a different base changing operation has been applied to System 3 than that applied to System 2. Similarly, System 6 and System 7 are the same but have had different base changing operations applied. For each system, Figure \ref{fig:10} compares statistics obtained from the $\efe$ observations, made before the change of base, with statistics obtained from the $\efe$ observations made after the change of base.
\begin{figure}[!t]
\centering
\includegraphics[width=0.75\textwidth]{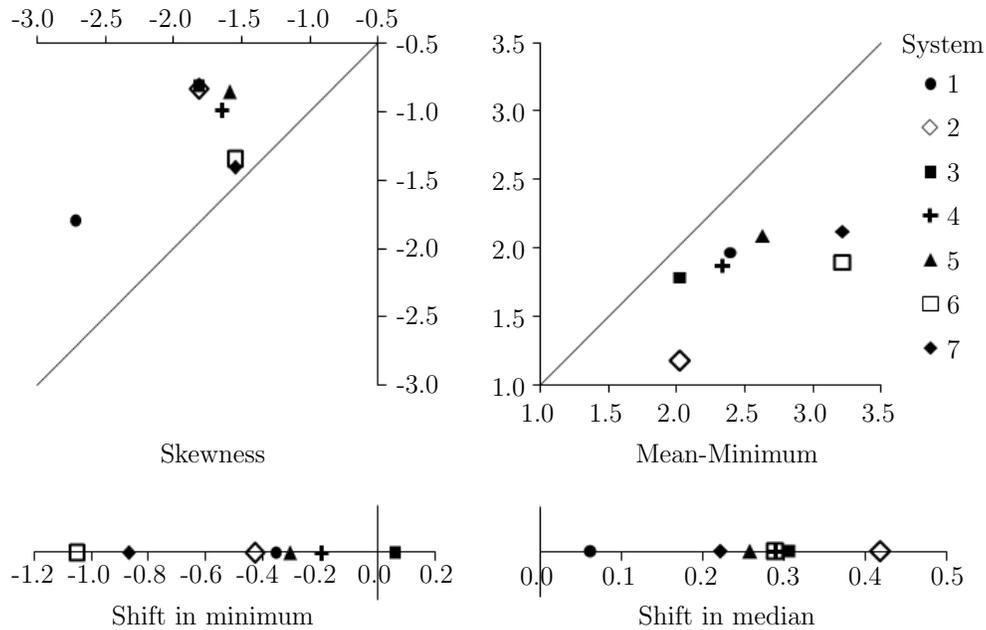}
\captionsetup{justification=centering,margin=1cm}
\caption{For each system the figure shows: the skewness, and mean minus minimum, of the $\efe$-histogram when using the original base ($x$ {\em axis}) and after changing to the alternative base ($y$ {\em axis}); the shift in the minimum and the shift in the median when changing back to the original base from the alternative base.}
\label{fig:10}
\end{figure}
Note that, in Figure \ref{fig:10}, skewness is measured using the adjusted Fisher-Pearson standardized moment coefficient.\\
For each of the systems investigated it can be seen from Figure \ref{fig:10} that, when changing back to the original base from the alternative base, the $\efe$-histogram undergoes an increase in negative skewness and mean minus minimum as well as a right shift in location. Furthermore, for most of the systems, the minimum $\efe$ value observed, when using the original base, is to the left of the minimum $\efe$ value observed when using the alternative base. These results suggest that the bases maximising the length of the left tail of the $\efe$-distribution (here approximated by an $\efe$-histogram) are important for the theory presented in the present paper. One caveat concerning this investigation is that, for each system, the variance in the observed minimum might be rather high because the distribution being sampled has a very thin left tail. We now move onto our next investigation.
\subsection{Joining and partitioning systems}
\label{subsec:3.2}
Consider the visual cortex and the auditory cortex of the brain. There is evidence that the brain defines relationships between the states of these different brain regions at a high level of meaning (i.e.\ between images and sounds); see \cite{Shinkareva}. However, at the lower level of meaning at which the images and the sounds are defined it may be that the two brain regions are self contained. The brains of two different people is perhaps a more overt example of self containment or privacy.\\
In the context of the theory of the present paper, suppose we have two systems. We can solve (\ref{equ:minefe}) for both of the systems separately and sum the resulting minimised $\efe$s. If this sum is significantly more than the minimum $\efe$ obtained when joining the two systems then it makes sense to consider the two systems as a single system. Examples of such systems are easy to construct. Conversely, for a given system, it might be possible to partition the set of nodes $S$ such that the sum of the minimum $\efe$s of the resulting systems is less than that of the original system. In this case, at least in the given basis, it makes sense to consider the original system as several different systems. It is not so easy to find examples of such systems, at least when the systems are small. However, Table \ref{tab:12} provides an example where the minimum $\efe$ of the system is greater than 3 whilst, after partitioning, the sum of the minimum $\efe$s is only 2.8.
\begin{table}[h]
\small
\centering
\captionsetup{justification=centering,margin=2cm}
\caption{Typical data of a system before and after a partition which results in lowering the total minimum $\efe$.}\smallskip
\begin{tabular}{c|cccc}
              &  node 1&  node 2&  node 3&  node 4\\ \hline
 $S_{\tau(1)}$& $v_{1}$& $v_{1}$& $v_{1}$& $v_{1}$\\
 $S_{\tau(2)}$& $v_{1}$& $v_{1}$& $v_{1}$& $v_{3}$\\
 $S_{\tau(3)}$& $v_{1}$& $v_{1}$& $v_{3}$& $v_{1}$\\
 $S_{\tau(4)}$& $v_{1}$& $v_{1}$& $v_{3}$& $v_{3}$\\
 $S_{\tau(5)}$& $v_{1}$& $v_{1}$& $v_{3}$& $v_{3}$\\
 $S_{\tau(6)}$& $v_{1}$& $v_{2}$& $v_{1}$& $v_{1}$\\
 $S_{\tau(7)}$& $v_{1}$& $v_{2}$& $v_{1}$& $v_{3}$\\
 $S_{\tau(8)}$& $v_{1}$& $v_{2}$& $v_{3}$& $v_{1}$\\
 $S_{\tau(9)}$& $v_{1}$& $v_{2}$& $v_{3}$& $v_{3}$\\
$S_{\tau(10)}$& $v_{1}$& $v_{2}$& $v_{3}$& $v_{3}$\\
$S_{\tau(11)}$& $v_{2}$& $v_{1}$& $v_{1}$& $v_{1}$\\
$S_{\tau(12)}$& $v_{2}$& $v_{1}$& $v_{1}$& $v_{3}$\\
$S_{\tau(13)}$& $v_{2}$& $v_{1}$& $v_{3}$& $v_{1}$\\
$S_{\tau(14)}$& $v_{2}$& $v_{1}$& $v_{3}$& $v_{3}$\\
$S_{\tau(15)}$& $v_{2}$& $v_{1}$& $v_{3}$& $v_{3}$\\
$S_{\tau(16)}$& $v_{2}$& $v_{1}$& $v_{1}$& $v_{1}$\\
$S_{\tau(17)}$& $v_{2}$& $v_{1}$& $v_{1}$& $v_{3}$\\
$S_{\tau(18)}$& $v_{2}$& $v_{1}$& $v_{3}$& $v_{1}$\\
$S_{\tau(19)}$& $v_{2}$& $v_{1}$& $v_{3}$& $v_{3}$\\
$S_{\tau(20)}$& $v_{2}$& $v_{1}$& $v_{3}$& $v_{3}$\\
$S_{\tau(21)}$& $v_{2}$& $v_{2}$& $v_{1}$& $v_{1}$\\
$S_{\tau(22)}$& $v_{2}$& $v_{2}$& $v_{1}$& $v_{3}$\\
$S_{\tau(23)}$& $v_{2}$& $v_{2}$& $v_{3}$& $v_{1}$\\
$S_{\tau(24)}$& $v_{2}$& $v_{2}$& $v_{3}$& $v_{3}$\\
$S_{\tau(25)}$& $v_{2}$& $v_{2}$& $v_{3}$& $v_{3}$
\end{tabular}
\quad
\begin{tabular}{c|cc}
              &  node 1&  node 2\\ \hline
$S'_{\tau(1)}$& $v_{1}$& $v_{1}$\\
$S'_{\tau(2)}$& $v_{1}$& $v_{2}$\\
$S'_{\tau(3)}$& $v_{2}$& $v_{1}$\\
$S'_{\tau(4)}$& $v_{2}$& $v_{1}$\\
$S'_{\tau(5)}$& $v_{2}$& $v_{2}$
\end{tabular}
\qquad
\begin{tabular}{c|cc}
               &  node 3&  node 4\\ \hline
$S''_{\tau(1)}$& $v_{1}$& $v_{1}$\\
$S''_{\tau(2)}$& $v_{1}$& $v_{3}$\\
$S''_{\tau(3)}$& $v_{3}$& $v_{1}$\\
$S''_{\tau(4)}$& $v_{3}$& $v_{3}$\\
$S''_{\tau(5)}$& $v_{3}$& $v_{3}$
\end{tabular}
\label{tab:12}
\end{table}
The result was obtained from the system by investigating an $\efe$-histogram involving $4\cdot 10^{6}$ observations, and by running the binary search algorithm. Note that the typical data of the system is such that the partitioned systems are independent when considered as random variables (this is why the number of typical data elements can be divided by five after partitioning).\\
The number of different partitions that there are of a system with $n$ nodes is given by the Bell number $B_{n}$. For $\#S=n$ we have $B_{n}=B^{(n)}(0)$, where $B(x)=\exp(e^{x}-1)$ is the generating function for the Bell number; see \cite{Mansour}. This number quickly becomes large as $n$ increases; making the investigation of all the different partitions of a system computationally expensive for all but small systems. In the final investigation of this section we consider the metric used in Definition \ref{def:fe}.
\subsection{Metric independence}
\label{subsec:3.3}
Remark \ref{rem:fe} suggests that the theory presented in the present paper is independent of the choice of metric used in Definition~\ref{def:fe} provided that the metric determines a total order on $[0,1]$ in some natural way. Before considering this in more detail, we have the following example.
\begin{example}
\label{exa:Output_22}
This example uses the typical data $T$ that was obtained in Example \ref{exa:Output_5.1}. The binary search algorithm of Subsection \ref{subsec:2.2} was again applied to $T$ but this time $\dis_{\infty}$ was used in Definition~\ref{def:fe} instead of $\dis_{1}$. After four cycles, the weighted relations in Table \ref{tab:13} were returned with $\efe(R,U,T)=5.30610$ using $\dis_{\infty}$.
\begin{table}[ht]
\small
\centering
\caption{Approximate solution for Example \ref{exa:Output_22} using $\dis_{\infty}$ in Definition~\ref{def:fe}.}\smallskip
\begin{tabular}{c|cccc}
    $U$&      0&147.224&294.449&441.673\\ \hline
      0&      1&0.53125&0.28125&0.03125\\
147.224&0.53125&      1&0.65625&0.34375\\
294.449&0.28125&0.65625&      1&0.59375\\
441.673&0.03125&0.34375&0.59375&1
\end{tabular}\\
\vspace{3mm}
\begin{tabular}{c|ccccc}
   $R$& node 1& node 2& node 3& node 4& node 5\\ \hline
node 1&      1&0.96875&0.84375&0.90625&0.90625\\
node 2&0.96875&      1&0.84375&0.90625&0.96875\\
node 3&0.84375&0.84375&      1&0.84375&0.84375\\
node 4&0.90625&0.90625&0.84375&      1&0.90625\\
node 5&0.90625&0.96875&0.84375&0.90625&1
\end{tabular}
\label{tab:13}
\end{table}
We see that $U$ in Table \ref{tab:13} yield the same graph illustration as that given by $U$ in Table \ref{tab:5} from Example \ref{exa:Output_5.1}. The same cannot be said when comparing $R$ in Table~\ref{tab:13} with $R$ in Table \ref{tab:5}, although there are some similarities. However, it also turns out that $U$ and $R$ in Table \ref{tab:5} are a better approximate solution to (\ref{equ:minefe}) in the present example, i.e. when using $\dis_{\infty}$ instead of $\dis_{1}$, than that given by $U$ and $R$ in Table \ref{tab:13}. Indeed the $\efe$ drops from $5.30610$ to $5.27370$ bpe.
\end{example}
The result in Example \ref{exa:Output_22} is perhaps not surprising once we appreciate certain similarities between $\dis_{\infty}$ and $\dis_{1}$. To appreciate these similarities and further results, the following assumption will be useful.
\begin{assumption}
\label{assumpt:metric_projection}
Under this assumption, for a metric $\dis:[0,1]^{n}\to\mathbb{R}_{+}$, there exists a metric $\dis':[0,1]\to\mathbb{R}_{+}$ such that for all $1\leq i\leq n$, $(c_{1},\cdots,c_{i-1},c_{i+1},\cdots,c_{n})\in[0,1]^{n-1}$ and $a,b\in[0,1]$ we have
\begin{gather*}
\dis((c_{1},\cdots,c_{i-1},a,c_{i+1},\cdots,c_{n}),(c_{1},\cdots,c_{i-1},b,c_{i+1},\cdots,c_{n}))=\dis'(a,b).
\end{gather*}
Furthermore, there is $a_{\min}\in[0,1]$ (e.g.\ classically $a_{\min}=0$) such that $\leq_{\dis}$, given by
\begin{gather}
\label{equ:TO}
a\leq_{\dis}b\Leftrightarrow\dis'(a_{\min},a)\leq\dis'(a_{\min},b),
\end{gather}
is a total order on $[0,1]$ and (up to reverse ordering) no other choice of $a_{\min}\in[0,1]$ in (\ref{equ:TO}) gives a different total order. Moreover, $\leq_{\dis}$ determines a maximum element $a_{\max}\in[0,1]$ (e.g.\ classically $a_{\max}=1$) which, if using $\dis$ in Definition~\ref{def:fe} and $\leq_{\dis}$ in the interpretation of weighted relations, would be used in the definition of a reflexive weighted relation in Definition~\ref{def:WR}.
\end{assumption}
\begin{remark}
\label{rem:metric_projection}
\begin{enumerate}
Assumption \ref{assumpt:metric_projection} gives rise to the following remarks.
\item[1.]
Under Assumption \ref{assumpt:metric_projection} it can be argued that $\dis$ determines a single well defined metric on $[0,1]$ and that this metric is~$\dis'$.
\item[2.]
Furthermore, $\dis'$ (and hence $\dis$) determines intervals in $[0,1]$, i.e.
\begin{gather*}
[a,b]_{\dis}:=\{c\in[0,1]:a\leq_{\dis}c\leq_{\dis}b\}\quad\mbox{for }a,b\in[0,1],
\end{gather*}
and the length of such intervals is given by $\dis'(a,b)$.
\item[3.]
With the above two remarks in place, we note that, in each coordinate, $\dis$ defines $\dis$-uniform random variables on $[0,1]$, i.e.\ if $[a,b]_{\dis}$ and $[c,d]_{\dis}$ are of the same length then the probability of a $\dis$-uniform random variable taking a value in $[a,b]_{\dis}$ is the same as it taking a value in $[c,d]_{\dis}$.
\item[4.]
To appreciate some of the similarities between $\dis_{\infty}$ and $\dis_{1}$, note that all of the metrics given in Definition \ref{def:d1} satisfy Assumption \ref{assumpt:metric_projection} and that $\dis'_{1}(a,b)=|a-b|=\dis'_{\infty}(a,b)$ for all $a,b\in[0,1]$. There are also some important differences between $\dis_{1}$ and $\dis_{\infty}$; we will look at some of these shortly.
\end{enumerate}
\end{remark}
We now return to the issue of metric independence.
\begin{lemma}
\label{lem:Independence}
Suppose Definition~\ref{def:fe} uses a metric $\dis$ that satisfies Assumption \ref{assumpt:metric_projection}, and let $f:[0,1]\to[0,1]$ be a bijection. Then $\dis_{f}:[0,1]^{n}\to\mathbb{R}_{+}$,
\begin{gather*}
\dis_{f}((a_{1},\cdots,a_{n}),(b_{1},\cdots,b_{n})):=\dis((f(a_{1}),\cdots,f(a_{n})),(f(b_{1}),\cdots,f(b_{n}))),
\end{gather*}
is also a metric on $[0,1]^{n}$ and the theory in the present paper is independent of a change of metric from $\dis$ to $\dis_{f}$ in Definition~\ref{def:fe} provided that, in Definition~\ref{def:WR}, $f^{-1}(a_{\max})$ is used in the definition of a reflexive weighted relation, $\leq_{\dis_{f}}$ is used in place of $\leq_{\dis}$ in the interpretation of values given by weighted relations and, when obtaining $\efe$-histograms (see Subsection \ref{subsec:2.3}), $\dis_{f}$-uniform random variables are used instead of $\dis$-uniform random variables.
\end{lemma}
\begin{proof}
Lemma \ref{lem:Independence} follows immediately from the fact that $\dis_{f}$ is merely $\dis$ under relabeling each $a\in[0,1]$ with $f^{-1}(a)$.
\end{proof}
It is expected that a more general result than Lemma \ref{lem:Independence} is possible such that $\dis'$ in Assumption \ref{assumpt:metric_projection} may have dependence on $(c_{1},\cdots,c_{i-1},c_{i+1},\cdots,c_{n})\in[0,1]^{n-1}$ for $1\leq i\leq n$. In this case, the interpretation of each value in a weighted relation table would be dependent on the other values in that table; the interpretation itself is determined by the metric $\dis$ being used.\\
To appreciate one of the differences between $\dis_{1}$ and $\dis_{\infty}$ we require the following definition.
\begin{definition}
\label{def:IFCD}
Let $\dis:[0,1]^{n}\to\mathbb{R}_{+}$ be a metric conforming to Assumption \ref{assumpt:metric_projection}. Moreover, let $a,b,c\in[0,1]^{n}$ be such that $\dis'(a_{i},c_{i})\geq\dis'(b_{i},c_{i})$, for $i=1,\ldots,n$, and for one or more $i$ we have $\dis'(a_{i},c_{i})>\dis'(b_{i},c_{i})$. If for all such $a,b,c\in[0,1]^{n}$ we have $\dis(a,c)>\dis(b,c)$ then $\dis$ is called an {\em increasing function of coordinatewise distance}.
\end{definition}
The metric $\dis_{1}$ is an increasing function of coordinatewise distance but, for $n>1$ in Definition \ref{def:IFCD}, $\dis_{\infty}$ is not; indeed $\dis_{\infty}(a,c)=1=\dis_{\infty}(b,c)$ for $n=2$, $a=(1,0.5)$, $b=(1,0)$ and $c=(0,0)$.\\
It is hoped that, upon further investigation, a class of metrics will emerge as being the most optimal (in some natural way) in the context of the theory of the present paper. Hence, independence arguments would then only need to apply to this class of metrics. It may be that being an increasing function of coordinatewise distance is a necessary condition for a metric to be optimal, but this is for future work. Regarding the theory in the present paper, it is certainly the case that the meaning of the values in weighted relation tables is given by the characteristics of the metric being used in Definition \ref{def:fe}.\\
We conclude this section with a reminder of the variety of different metrics that there are on $\mathbb{R}^{n}$. Lemma \ref{lem:em} shows that, even when restricting attention to metrics that are equivalent to $\dis_{2}$, the variety is great.
\begin{lemma}
\label{lem:em}
Let $\dis_{2}$ be the Euclidean metric on $\mathbb{R}^{n}$, $n\in\mathbb{N}$. For all $a,b\in\mathbb{R}^{n}$, define $\dis^{f}(a,b):=\dis_{2}(f(a),f(b))$, where $f:\mathbb{R}^{n}\to\mathbb{R}^{n}$ is such that $f:(\mathbb{R}^{n},\dis_{2})\to(\mathbb{R}^{n},\dis_{2})$ is a homeomorphism. Then $(\mathbb{R}^{n},\dis^{f})$ is a metric space and $\dis^{f}$ is equivalent to $\dis_{2}$; i.e.\ the open subsets of $(\mathbb{R}^{n},\dis^{f})$ are the same as those of $(\mathbb{R}^{n},\dis_{2})$.
\end{lemma}
\begin{proof}
Since $f:(\mathbb{R}^{n},\dis_{2})\to(\mathbb{R}^{n},\dis_{2})$ is a homeomorphism, $f:\mathbb{R}^{n}\to\mathbb{R}^{n}$ is a bijection. From this it easily follows that $\dis^{f}$ is a metric. To show equivalence we need to show that $A\subseteq(\mathbb{R}^{n},\dis^{f})$ is open if and only if $A\subseteq(\mathbb{R}^{n},\dis_{2})$ is open. Hence, to show only if, let $A\subseteq(\mathbb{R}^{n},\dis^{f})$ be open. In this direction it is enough to show that $f(A)$, as a subset of $(\mathbb{R}^{n},\dis_{2})$, is open since then $A=f^{-1}(f(A))\subseteq(\mathbb{R}^{n},\dis_{2})$ is open by $f:(\mathbb{R}^{n},\dis_{2})\to(\mathbb{R}^{n},\dis_{2})$ being a homeomorphism. Let $a'\in f(A)$. Then $a'=f(a)$ for some $a\in A$. Since $A\subseteq(\mathbb{R}^{n},\dis^{f})$ is open, there exists $\varepsilon>0$ such that for all $b\in\mathbb{R}^{n}$ with $\dis^{f}(a,b)<\varepsilon$ we have $b\in A$, and thus $f(b)\in f(A)$. Hence, for all $b'\in\mathbb{R}^{n}$ with $\dis^{f}(a,f^{-1}(b'))<\varepsilon$ we have $f^{-1}(b')\in A$, and thus $f(f^{-1}(b'))=b'\in f(A)$. Noting that $\dis^{f}(a,f^{-1}(b'))=\dis_{2}(f(a),f(f^{-1}(b')))=\dis_{2}(a',b')$, it follows from the last statement that for all $b'\in\mathbb{R}^{n}$ with $\dis_{2}(a',b')<\varepsilon$ we have $b'\in f(A)$. Hence, $f(A)\subseteq(\mathbb{R}^{n},\dis_{2})$ is open. The proof in the other direction is similar.
\end{proof}
\begin{example}
\label{exa:em}
Let $f:\mathbb{R}^{2}\to\mathbb{R}^{2}$ conform to the conditions of Lemma \ref{lem:em}. If $f$ maps $\ell$ in Figure \ref{fig:11} to the unit circle in $(\mathbb{R}^{2},\dis_{2})$ and $f(0)=0$ then $\ell\subseteq\mathbb{R}^{2}$ is the unit circle in $(\mathbb{R}^{2},\dis^{f})$.
\begin{figure}[!t]
\centering
\includegraphics[width=0.30\textwidth]{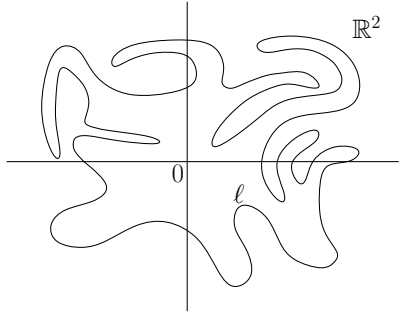}
\captionsetup{justification=centering,margin=2cm}
\caption{The path $\ell$ in Example \ref{exa:em}.}
\label{fig:11}
\end{figure}
\end{example}
Section \ref{sec:4} provides some generalisations of Definition \ref{def:fe}, a comparison with both Integrated Information Theory and Shannon entropy, followed by the conclusion.
\section{Generalisation, a comparison with IIT theory, and conclusion}
\label{sec:4}
In Subsection \ref{subsec:4.1} we extend Definition \ref{def:fe} to involve many more weighted relations.
\subsection{Multi-relational float entropy and time dependents}
\label{subsec:4.1}
We start with a definition. 
\begin{definition}[Multi-relational float entropy]
\label{def:mrfe}
Let $S$ be as in Definition \ref{def:S}, let $U\in\Psi_{V}$, and let $R\in\Psi_{S}$. Furthermore, let $U_{1},U_{2},\ldots$ and $R_{1},R_{2},\ldots$ be weighted relations analogous to $U$ and $R$ but for the system in different bases; see Subsection~\ref{subsubsec:BBS} on base branching structure. The {\em multi-relational float entropy} of a data element $S_{i}\in\Omega_{S,V}$, relative to $U,U_{1},U_{2},\ldots$ and $R,R_{1},R_{2},\ldots$, is defined as
\begin{align*}
&\fe(R,U,R_{1},U_{1},R_{2},U_{2},\ldots,S_{i})\\
&\quad:=\log_{2}(\#\{S_{j}\in\Omega_{S,V}\colon\mbox{C}_{0}(R,U,R_{1},U_{1},R_{2},U_{2},\ldots,S_{i},S_{j})\wedge\mbox{C}_{1}(R,U,R_{1},U_{1},R_{2},U_{2},\ldots,S_{i},S_{j})\wedge\cdots\}),
\end{align*}
where the first condition $\mbox{C}_{0}(R,U,R_{1},U_{1},R_{2},U_{2},\ldots,S_{i},S_{j})$ is $\dis(R,R\{U,S_{j}\})\leq \dis(R,R\{U,S_{i}\})$, as in Definition \ref{def:fe}.
\end{definition}
In Definition \ref{def:mrfe}, all of the conditions $\mbox{C}_{0}, \mbox{C}_{1}, \cdots$ need to be satisfied for a data element $S_{j}$ to contribute toward the multi-relational float entropy of a data element $S_{i}$. The additional conditions should be those that increase the length of the left tail of the $\efe$-distribution.\\
For example, for some given system (and under the requirement of minimising expected multi-relational float entropy), $R$ might be such that it define two dimensional geometry on the nodes of the system. Furthermore, for a particular branch of the base branching structure, the states of the composite nodes will be images under the geometry (given by $R$) on the nodes that have been combined. Hence, for $\mbox{C}_{1}$ analogous to $\mbox{C}_{0}$ but using the new base, the system defines a weighted relation $U_{1}$ on the repertoire of the composite nodes and hence on a set of images. This may have relevance to some aspects of the Gestalt theory of visual perception; see \cite{Wagemans}.\\
Suppose that the system is part of the brain. As suggested in Subsection \ref{subsec:3.2}, at the level of meaning at which images are defined by the visual cortex and sounds are defined by the auditory cortex, it may be that the two brain regions are self contained; i.e.\ they may be separate systems. However, at a higher level of meaning (and for a particular branch of the base branching structure), one of the nodes of the brain will be the whole visual cortex and another will be the whole auditory cortex. The states of the visual cortex are visual objects and the states of the auditory cortex are sounds. Applying the present theory appropriately should give a weighted relation on the two cortical regions and another giving relationship values between objects and sounds. One caveat, however, is that in this case the two cortical regions as nodes do not share the same node repertoire, and so some care needs to be taken when considering how to apply the definitions of the present paper.\\
There is also evidence of sparse coding in various cortical regions; see \cite{Quiroga} and \cite{Graham}. For example, there are neurons that are active if and only if activity related to a specific object (auditory or visual etc) is present in the respective cortex. Hence, under the minimisation of $\efe$ also on this set of neurons, the system defines relationships between objects.\\
Finally, Definition \ref{def:fe} allows time dependent versions of the results presented in the present paper, and in general. Suppose in Example \ref{exa:Output_5.1} that the digital photographs sampled are in fact frames from videos. Choose an integer $k\in\mathbb{N}$. For each sampled frame, sample in the same way the subsequent $k-1$ frames so that the number of nodes of the system has increase by a factor of $k$ (i.e.\ each node in each typical data element is replaced by $k$ nodes that form a sequence of states of the original node over a short time period). In this case, it is anticipated that if $U$ and $R$ solve (\ref{equ:minefe}) then $R$ will define geometry on the nodes of the system that has a dimension for time.
\subsection{A comparison with Integrated Information Theory and Shannon entropy}
\label{subsec:IIT}
This subsection starts with an initial comparison between the theory of the present paper and Giulio Tononi's Integrated Information Theory (IIT) of consciousness. IIT has gained much attention in recent years (see \cite{Tononi}, \cite{Oizumi}, \cite{Balduzzi}, \cite{Balduzzi2} and \cite{Barrett}), and it maybe that the two theories are quite compatible in some areas. There is a significant difference in emphasis in the formulation of the two theories. In \cite{Oizumi} IIT has been formulated and further developed with the intention that it will satisfy certain self-evident truths about consciousness, which Tononi refers to as axioms. In brief, the axioms are as follows:
\begin{itemize}
\item Existence: Consciousness exists.
\item Composition: Within the same experience, one can see, for example, left and right, red and green, a circle and a square, a red circle on the left, a green square on the right, and so on.
\item Information: Consciousness is informative: each experience differs in its particular way from an immense number of other possible experiences.
\item Integration: Each experience is irreducible to independent components.
\item Exclusion: At any given time there is only one experience having its full content. This axiom also states constraints on consciousness such as resolution.
\end{itemize}
From these axioms IIT postulates a number of properties that physical systems must satisfy in order to generate consciousness. These properties introduce a substantial amount of initial theory involving cause and effect within systems. Since this initial theory is fundamental to the formulation of IIT, it is crucial that the set of axioms is correct and complete. The theory of the present paper has a significant difference in emphasis because it uses the following axiom in its formulation:
\begin{itemize}
\item Relations: Consciousness is awash with underlying relationships which provide the relational content of experience.
\end{itemize}
It is natural that relations should be fundamental to the formulation of a theory of consciousness because, in one form or another, they are ubiquitous among mathematical structures. Hence, in the author's opinion, this axiom should be added to Tononi's list of axioms. However, IIT does have something to say about the quality of conscious experience, and this is discussed below.\\
It is worth noting that the theory in the present paper is more or less compatible with the IIT axioms. For example, regarding the unity of consciousness (integration), according to the theory in the present paper, when a brain state is interpreted in the context of the weighted relations that minimise expected multi-relational float entropy, the brain state acquires meaning in the form of the relational content of the associated experience. Furthermore, regarding resolution (part of the exclusion axiom) we recall that, for all but trivial examples, (\ref{equ:minefe}) will have many solutions and, hence, only defines weighted relations up to a certain resolution that depends on the system. Of course a more rigorous comparison with the axioms is desirable, but this is for future work.
\subsubsection{Mechanisms contributing to consciousness}
\label{subsubsec:Mech}
One of the strengths of IIT is that it attempts to distinguish between brain regions that contribute toward consciousness and those that do not. This is undertaken at several different scales from small mechanisms (i.e. small subsystems) up to whole systems such as the brain. For this purpose, at the scale of mechanisms, the theory introduces a quantity called {\em Integrated Information}, and analogous quantities are introduced for larger scales. It is worth giving the reader some insight into how this quantity is defined for mechanisms. Suppose we have a small number of logic gates that are interconnected in some way, and that the resulting mechanism updates over discrete time. The current state of the mechanism (say at time $t=0$) provides causal information about what the state of the mechanism might have been at time $t=-1$. In fact it implies a probability distribution on the set of all states of the mechanism for $t=-1$. If we were to partition the mechanism in some way by cutting connections and treating cut inputs to gates as extrinsic noise then, in many cases, there would be a reduction in the amount of causal information that the current state of the mechanism provides about what the state of the mechanism might have been at $t=-1$. As in the case of the unpartitioned mechanism, the partitioned mechanism also implies a probability distribution on the set of all states of the mechanism for $t=-1$. The reduction in the causal information is quantified by measuring the distance between these two probability distributions using the Wasserstein metric, also known as the earth-mover's distance.\\
If, out of all the different ways to partition the mechanism, the partition chosen actually loses the minimum amount of causal information then the partition is called the minimum information partition (MIP) for the mechanism in its given state at $t=0$.\\
In IIT, the quantity of {\em causal information} of a mechanism in its given state is defined as the Wasserstein distance between the probability distribution for the state of the unpartitioned mechanism, at $t=-1$, and the probability distribution for the state of the mechanism's MIP at $t=-1$. In IIT there is also an analogous definition for the quantity of {\em effect information} of a mechanism in its given state. Finally, the quantity of {\em integrated information} of a mechanism in its given state is defined as the minimum of its causal information and its effect information for that state. The integration postulate of IIT says that only when the quantity of integrated information is positive can a mechanism contribute to consciousness.\\
We will now consider an example from \cite{Balduzzi} which formed part of the motivation behind the definition of integrated information. We will see that there is an alternative (or complimentary) interpretation of the example which leads in the direction of the theory of the present paper. Consider a digital-camera sensor chip made up of 1 million photodiodes. From the perspective of an external observer, the chip has a large number of different states. From an intrinsic perspective, however, the chip can be considered as 1 million independent photodiodes; cutting the chip down into individual photodiodes would not change the performance of the camera. It is hard to imagine that the chip can be conscious of the images that fall upon it. On the other hand, the visual experiences we enjoy are integrated and we experience whole images. Accordingly, cutting the visual cortex down into individual neurones would completely change the performance of the system.\\
It is then stated in the example that what underlies the unity of experience is a multitude of causal interactions among the relevant parts of the brain. From this we can see why cause and effect is a fundamental part of the definitions used in IIT, and why the theory developed in the direction it did. An alternative (or complimentary) interpretation of the example is that the interactions between neurons make some states of the system more likely than other states; i.e. the system is inherently biased and this defines a probability distribution $P$ on the set of states of the system. The probability distribution is a property of the system itself and allows the system to define expected quantities. This allows the theory of the present paper to be developed with an emphasis on relations, which is desirable since relationships are an inherent part of consciousness.\\
Now let's consider the camera chip in the context of the theory of the present paper. Each photodiode is unbiased since its state is driven by its input signal. The 1 million photodiodes are completely independent. If the chip defines a probability distribution on its states at all (which is debatable) then it is the uniform distribution. In Examples \ref{exa:Output_5.4} and \ref{exa:Output_5.6} of the present paper, we saw that when $P$ is uniform the solutions to (\ref{equ:minefe}) vary greatly and, hence, the system fails to define weighted relations that give a coherent interpretation of the states of the system. Furthermore, the associated $\efe$-histogram is without a left tail. So, for contrast with IIT, the theory of the present paper suggests that, to contribute to consciousness, a mechanism will at least need an inherent probability distribution on its set of states that gives an $\efe$-histogram with a long left tail. The length of the left tail may turn out to be of great importance; when the tail is very long, the solutions to (\ref{equ:minefe}) are very distinct from other weighted relations. The length of the left tail is also important in multi-relational float entropy regarding which branches of the base branching structure should be involved; see Subsection \ref{subsec:4.1}.\\ 
From a practical perspective, we might use cause and effect to estimate the inherent probability distribution of a mechanism. For a deterministic mechanism, we can estimate the probability of a state $S_{i}$ as the number of states that cause $S_{i}$ divided by the total number of states of the mechanism. Of course, Markov Chain theory is appropriate here (particularly in the nondeterministic case) and a rigorous approach should be taken.
\subsubsection{The quality of consciousness}
\label{subsubsec:Quality}
Suppose we have a mechanism that has $n$ states. In IIT (see \cite{Balduzzi}), the {\em Qualia space} of the mechanism is an $n$-dimensional space with a real axis given for each state of the mechanism. Each probability distribution on the set of states of the mechanism defines a point in an $(n-1)$-dimensional subspace of Q-space, noting that, for each probability distribution, probabilities must sum to 1. The point closest to the origin in this subspace is given by the uniform distribution. For a given state of the mechanism at $t=0$, the state defines a probability distribution on the set of states of the mechanism at $t=-1$ and, hence, defines a point in Q-space. Similarly, for the given state, each partitioned version of the mechanism (i.e. where only a subset of the set of connections of the mechanism is present) also defines a point in Q-space. In IIT, some of these points in Q-space are joined by {\em q-arrows}; the connections of the mechanism involved in determining the point at the bottom of a q-arrow are included in the subset of connections involved in determining the point at the top of the q-arrow. This forms a lattice, embedded in Q-space, which has the uniform distribution at its bottom and the distribution given by the complete mechanism at its top. The shape that the lattice encloses is called the {\em quale}, and the q-arrows are a geometric realization of information relationships.\\
Changing the state of the mechanism at $t=0$ will, usually, change the shape of the lattice embedded in Q-space. According to IIT, the shape completely specifies the quality of the experience, and it is suggested in \cite{Balduzzi} that similarity in shape corresponds to similarity in experience. The theory in \cite{Balduzzi} also suggests a way in which relationships, giving the geometry of monocular vision, might be defined in Q-space, although the theory has not been developed in a way that prioritises a capability for defining relationships. The property involved concerns q-arrows and is referred to, in \cite{Balduzzi}, as {\em entanglement}. Suppose a lattice in Q-space has a point $p_{1}$ that is at the bottom of two q-arrows $q_{1,2}$ and $q_{1,3}$ which terminate at points $p_{2}$ and $p_{3}$ respectively. The connections of the mechanism involved in determining the points $p_{2}$ and $p_{3}$, when taken together, determine a point $p_{4}$. Treating $p_{1}$ and $p_{4}$ as vector from the origin, if $p_{4}\not=p_{1}+q_{1,2}+q_{1,3}$ then the q-arrows $q_{1,2}$ and $q_{1,3}$ are said to be tangled. In other words, the information relationship given by $q_{1,4}$ does not reduce to the information relationships given by $q_{1,2}$ and $q_{1,3}$.\\
With respect to vision, it is suggested in \cite{Balduzzi} that, for a mechanism in the form of a grid, connections of the mechanism that are close together will give entangled q-arrows in Q-space near the bottom of the lattice, but connections of the mechanism that are far apart will not. Hence, these entanglements give rise to the concept of local regions and, therefore, geometry.\\
From the perspective of the author of the present paper, entanglement is a desirable addition to the theory of integrated information since it acknowledges the need for the theory to include the capacity to define relationships. For comparison regarding the quality of consciousness, the aim of the present paper is to provide a mathematical theory for how the brain defines the relationships underlying consciousness. If applicable to the visual cortex, the examples in Section \ref{sec:3} show that the perceived relationships between different colours, the perceived relationship between different brightnesses, and the perceived relationship between different points in a person's field of view (giving geometry) are all defined by the brain in a mutually dependent way. If we were to apply the theory to the auditory cortex then the resulting weighted relations might define how we perceive the relationships between the pitches of the chromatic scale. Of course, more work is required. Although an early example, when considering the scope of the theory of the present paper, readers may find the Definitive Player Problem to be of interest; see \cite{Mason}. In short, when IIT leans in the direction of defining relationships synergies start to emerge with the theory of the present paper.
\subsubsection{A comparison of float entropy and Shannon entropy}
\label{subsubsec:Shannon}
Shannon entropy is notably used in the neuroscience of consciousness. The definition of float entropy (see Definitions \ref{def:fe} and \ref{def:mrfe}) has some similarity to that of Boltzmann's entropy. Whilst not to be confused with Shannon entropy, expected float entropy, $\efe$, does have some similarities with Shannon entropy. Indeed, $\efe$ is a measure (in bits per data element) of the expected amount of information needed, to specify the state of the system, beyond what is already known about the system from the weighted relations provided. Shannon entropy is a measure of information content in data. As data becomes more random, Shannon entropy increases because structure in data is actually a form of redundancy. By solving (\ref{equ:minefe}) for a given system we obtain a structure in the form of weighted relations defined by the system. Relative to these weighted relations, if the system was to become more random then the $\efe$ value for the system would increase.
In order to make the similarities between $\efe$ and Shannon entropy clearer, consider the summation
\begin{gather}
\label{equ:Shannon}
\sum_{S_{i}\in\Omega_{S,V}}P(S_{i})\log_{2}\left(\frac{1}{P(S_{i}\mid A_{S_{i}})}\right),
\end{gather}
where $A_{S_{i}}:=\{S_{j}\in\Omega_{S,V}\colon \dis(R,R\{U,S_{j}\})\leq \dis(R,R\{U,S_{i}\})\}$. The summation in (\ref{equ:Shannon}) is similar in form to the definition of Shannon entropy. Furthermore, (\ref{equ:Shannon}) can be written as
\begin{gather}
\label{equ:Shannon2}
\sum_{S_{i}\in\Omega_{S,V}}P(S_{i})\log_{2}\left(\frac{\sum_{S_{j}\in A_{S_{i}}}P(S_{j})}{P(S_{i})}\right),
\end{gather}
and, when the probabilities in the argument of the logarithm are comparable, this will give a value similar to $\efe(R,U,P)$. Finally, we can write (\ref{equ:Shannon2}) as
\vspace{-0.5cm}
\begin{gather}
\label{equ:Shannon3}
H+\sum_{S_{i}\in\Omega_{S,V}}P(S_{i})\log_{2}\left(\sum_{S_{j}\in A_{S_{i}}}P(S_{j}) \right),
\end{gather}
where $H$ is the Shannon entropy of the system and, with consideration of the log function, the second term has a negative value between $-H$ and 0. As per Example \ref{exa:Output_5.6}, even when $P$ is uniform over $\Omega_{S,V}$, the second term of (\ref{equ:Shannon3}) need not be equal to 0. However, for $U$ and $R$ the constant functions which everywhere take the value 1, (\ref{equ:Shannon3}) simplifies to $H$.
\subsection{Conclusion}
\label{subsec:4.2}
The present paper significantly extends the work introduced in \cite{Mason} by further developing theory and testing theory using several informative examples. We have noted the following two facts. Firstly, conscious experience is awash with underlying relationships. Secondly, for various brain regions, such as the visual cortex, the probability distribution over the different possible states of the system is far from being uniform owing to the effect of learning rules that weaken or strengthen synapses. Hebb's principle says that what fires together wires together and the BCM version of Hebbian theory is one of many such learning paradigms; see \cite{BCM} and \cite{KRB}. There is also evidence for the relevance of BCM theory regarding the hippocampus; see~\cite{DB}. Furthermore, the probability distribution over the states of the system is a property of the system itself allowing the system to define expected quantities. The theory in the present paper provides a link between the above facts. Under the requirement of minimising expected (multi-relational) float entropy, the brain defines relationships; the theory represents relationships using weighted relations. It is proposed that when a brain state is interpreted in the context of all these weighted relations, defined by the brain, the brain state acquires meaning in the form of the relational content of the associated experience. The examples in the present paper provide evidence that supports the theory.\\
In Example \ref{exa:Output_5.1}, $T$ was obtained from digital photographs having a four shade gray scale. In this case, $T$ has defined the correct relationships under the requirement that $\efe$ is minimised. Similarly, in Example \ref{exa:Output_6.1}, $T$ was obtained from digital photographs having a nine colour red/green palette. We note that, given the system involved, $R$ and $U$ in this example also appear to be favorable weighted relations, and appropriate as approximate solutions to (\ref{equ:minefe}). However, in this case $R$ and $U$ were guessed and judged appropriate from the $\efe$-histogram; the actual solutions to (\ref{equ:minefe}) could be somewhat different.\\
The results in these examples suggest that the perceived relationships between different colours, the perceived relationships between different brightnesses, and the perceived relationships between different points in a person's field of view (giving geometry) are all defined by the brain in a mutually dependent way. Hence, in this case, there is a connection between the relationships that underly colour perception and our perception of the underlying geometry of the world around us.\\
If we were to apply the theory to the auditory cortex then the resulting weighted relations might define how we perceive the relationships between the pitches of the chromatic scale. Of course, more work is required. Although an early example, when considering the scope of the theory, readers may find the Definitive Player Problem to be of interest; see \cite{Mason}.\\
In Example \ref{exa:Output_5.6}, we applied the theory to a system where the probability distribution $P$ in Definition \ref{def:S} is uniform over $\Omega_{S,V}$. In this case the solutions to (\ref{equ:minefe}) vary greatly (instead of all being similar) and, hence, the system fails to define weighted relations that give a coherent interpretation of the states of the system. We found that the variation in the solutions to (\ref{equ:minefe}) is partly due to a type of symmetry within $T$; this is discussed in Example \ref{exa:Output_5.6}. Also, the associated $\efe$-histogram is without a left tail. This example supports the claim that the theory may satisfies the empirical observation that not all systems appear to be capable of consciousness.\\
In Subsection \ref{subsubsec:BEH}, we investigated the effect of applying base changing operations. Typical data was obtained for seven systems from digital photographs. For each of the systems investigated, Figure \ref{fig:10} shows that, when changing back to the original base from the alternative base, the $\efe$-histogram undergoes an increase in negative skewness and mean minus minimum as well as a right shift in location. Furthermore, for most of the systems, the minimum $\efe$ value observed, when using the original base, is to the left of the minimum $\efe$ value observed when using the alternative base. These results suggest that the bases maximising the length of the left tail of the $\efe$-distribution (here approximated by an $\efe$-histogram) are important for the theory presented in the present paper. However, instead of permitting all base changing operations, restricting the theory to the base branching structure may be necessary; see Subsection \ref{subsubsec:BBS}.\\
It is argued in \cite{Mason} that the theory presented there provides a solution to the binding problem and avoids the homunculus fallacy. Those arguments also apply to the theory presented in the present paper. In particular, consciousness is not the output of some algorithmic process but it may instead, largely, be the states of the system interpreted in the context of the weighted relations that minimise expected multi-relational float entropy; see Definition \ref{def:mrfe}. The weighted relations that Definition \ref{def:mrfe} involves, in addition to $U$ and $R$, are brought in to play by increasing the number of conditions in Definition  \ref{def:fe}. The extra conditions utilise higher bases of the base branching structure. The findings of the present paper suggest that the conditions $\mbox{C}_{0}, \mbox{C}_{1}, \cdots$ should be those that increase the length of the left tail of the $\efe$-distribution.\\
In Subsection \ref{subsec:3.2}, we investigated joining and partitioning systems. Table \ref{tab:12} provides an example where the minimum $\efe$ of the system is greater than 3 whilst, after partitioning, the sum of the minimum $\efe$s is only 2.8.\\
In Subsection \ref{subsec:3.3}, we considered whether the theory presented in the present paper is independent of the choice of metric used in Definition \ref{def:fe} when the metric determines a total order on $[0,1]$ in some natural way. In this case, the meaning of the values in weighted relation tables is determined by the metric being used. Example \ref{exa:Output_22} and Lemma \ref{lem:Independence} provide some evidence of such independence. However, some more work is required.\\
Finally, in Subsection \ref{subsec:IIT} we made some comparisons between the theory of the present paper, Integrated Information Theory, and Shannon entropy. The integration postulate of IIT says that only when the quantity of integrated information is positive can a mechanism contribute to consciousness. For comparison, the theory of the present paper suggests that, to contribute to consciousness, a mechanism will at least need an inherent probability distribution on its set of states that gives an $\efe$-histogram with a long left tail. The length of the left tail may turn out to be of great importance.\\
According to IIT, the shape of a quale in Q-space completely specifies the quality of the experience, and it is suggested in \cite{Balduzzi} that similarity in shape corresponds to similarity in experience. The theory in \cite{Balduzzi} also suggests a way in which relationships might be defined in Q-space by entangled q-arrows. For comparison, the theory of the present paper suggests that, under the requirement of minimising expected (multi-relational) float entropy, the brain defines relationships (represented in the theory by weighted relations) such that when a brain state is interpreted in the context of all these relationships the brain state acquires meaning in the form of the relational content of the associated experience.\\
Finally in Subsection \ref{subsubsec:Shannon} we showed that $\efe$ is a measure (in bits per data element) of the expected amount of information needed, to specify the state of the system, beyond what is already known about the system from the weighted relations provided.\\
It is hoped that future research will someday determine the extent to which the word `quasi' can be removed from the title of the present paper. Whilst rather different in content, readers may also find \cite{Edelman}, \cite{Ascoli}, \cite{Sporns} and \cite{Miyawaki} to be of interest.
\appendix
\section{Software}
\label{App:AppendixA}
\begin{table}[ht]
\small
\centering
\caption{Software used during the research for the present paper.}\smallskip
\begin{tabular}{lll}
Software&Availability&Use in the present paper\\ \hline
GIMP 2.6&Freeware&Used to posterise digital raster images, i.e. reduce the palette size to a small\\
&&number of shades or colours.\\
RasterSampler 1.0 (Java)&From the author&Used to sample pixels and collate data.\\
URFinder 3.7 (Java)&From the author&Used to search for solutions to (\ref{equ:minefe}) and collect observations for $\efe$-histogram.\\
Excel 2007&Microsoft&Used to generate binary entry tables (such as those in Table \ref{tab:4}), store outputs\\
&&and perform statistical analysis.\\
Minitab 17&Minitab Inc&Statistical analysis.
\end{tabular}
\label{tab:14}
\end{table}
URFinder 3.7 can be used to implement the binary search algorithm, specified in Subsection \ref{subsec:2.2}, and for collecting observations from which $\efe$-histograms can be produced. The author ran URFinder 3.7 on a desktop dual-core CPU machine, and is happy to distribute the software. The algorithm and machine were chosen for convenience and their performance (i.e. the maximum size of system that can practically be investigated) is far from what could potentially be achieved. Indeed, for a system with $n=\#S$ and $m=\#V$, Step 2 of the binary search algorithm calculates $2^{(n(n-1)+m(m-1))/2}$ exact $\efe$ values. This is computationally expensive for all but quite small systems, particularly since the algorithm calculates exact $\efe$ values rather than estimates obtained by employing statistical methods.\\
For future investigations we could consider taking advantage of the continuing increase in power and affordability of multi-GPU machines and hybrid CPU-GPU machines. The use of GPUs can result in orders of magnitude improvement in speed over conventional processors. Furthermore, (\ref{equ:minefe}) is an optimisation problem and falls within a common general class of problems studied in optimisation theory for which a number of efficient algorithms are available. These involve, gradient methods, stochastic gradient methods and derivative free optimisation; see 
\cite{Krejic}, \cite{Nocedal} and \cite{Conn}.
\section{Notation}
\label{App:AppendixB}
\begin{table}[ht]
\small
\centering
\caption{Notation (most of the formal definitions can be found in Subsections \ref{subsec:1.1}, \ref{subsec:3.3} and \ref{subsec:4.1}).}\smallskip
\begin{tabular}{rl}
Symbol&Description\\ \hline
$a,b,c,\ldots$&elements of $S$ but also used to denote elements of other sets when the meaning is clear from the\\
&context.\\
$A$&an element of $2^{\Omega_{S,V}}$.\\
$B_{n}$&the Bell number for $\# S=n$.\\
$B(x)=\exp(e^{x}-1)$&the generating function of $B_{n}$.\\
$C_{0}, C_{1}, C_{2},\ldots$&conditions, involving weighted relations, in the definition of multi-relational float entropy.\\
$\dis$&a metric on the set of all weighted relations on $S$ or, in places, a metric on $[0,1]^{n}$.
\end{tabular}
\label{tab:15}
\end{table}
\begin{table}[ht]
\small
\centering
\begin{tabular}{rl}
$\dis_{n}$&for $n\in\mathbb{N}\cup\{\infty\}$, a metric (on the set of all weighted relations on $S$) obtained from the corresponding\\
&$p$-norm, for $p=n$, on a finite dimensional vector space.\\
$\dis^{f}$&a metric on $\mathbb{R}^{n}$; a function $f:\mathbb{R}^{n}\to\mathbb{R}^{n}$ is used in its definition.\\
$\dis_{f}$&a metric on $[0,1]^n$; a function $f: [0,1]\to[0,1]$ is used in its definition.\\
$\leq_{\dis}$&a total order on $[0,1]$ determined by the metric $\dis$.\\
$[\cdot,\cdot]_{\dis}$&an interval determined by the metric $\dis$.\\
$\efe(R,U,P)$&the expected float entropy, relative to $U$ and $R$, of the given system.\\
$\efe(R,U,T)$&the mean approximation of $\efe(R,U,P)$.\\
$\fe(R,U,S_{i})$&the float entropy, relative to $U$ and $R$, of the data element $S_{i}$.\\
$\fe(R,U,R_{1},U_{1},R_{2},U_{2},\ldots,S_{i})$&the multi-relational float entropy, relative to $U,U_{1},U_{2},\ldots$ and $R,R_{1},R_{2},\ldots$, of the data element $S_{i}$.\\
$f_{i}$&the map $f_{i}:S\to V$ corresponding to the data element $S_{i}$.\\
node 1,node 2,node 3,$\ldots$&elements of $S$.\\
$P$&the probability distribution $P:\Omega_{S,V}\to[0,1]$ of the random variable defined by the bias of the given\\
&system. $P$ extends to a probability measure on $2^{\Omega_{S,V}}$.\\
$R$&an element of $\Psi_{S}$.\\
$R\{U,S_{i}\}$&the element of $\Psi_{S}$ given by the canonical definition $R\{U,S_{i}\}:=U(f_{i}(a),f_{i}(b))$ for all $a,b\in S$.\\
$S$&a nonempty finite set; in most places $S$ denotes the set of nodes of a system.\\
$S_{i}$&a data element for $S$, i.e. a system state given by the aggregate of the node states.\\
$T$&the typical data for the given system, i.e. $T$ is a finite set of numbered observations of the given\\
&system.\\
$\tau$&the map $\tau:\{1,\ldots,\# T\}\to\{i:S_{i}\in\Omega_{S,V}\}$ for which $S_{\tau(k)}$ is the value of observation number $k$ in $T$.\\
&$\tau$ need not be injective.\\
$U$&an element of $\Psi_{V}$.\\
$v_{1}, v_{2}, v_{3},\ldots$&elements of $V$.\\
$V$&the node repertoire, i.e. the set of node states for a given system.\\
$\Psi_{S}$&the set of all reflexive, symmetric weighted-relations on $S$.\\
$\Psi_{V}$&the set of all reflexive, symmetric weighted-relations on $V$.\\
$\Omega_{S,V}$&the set of all data elements $S_{i}$, given $S$ and $V$.\\
$2^{\Omega_{S,V}}$&the power set of $\Omega_{S,V}$.
\end{tabular}
\label{tab:15.1}
\end{table}
{\bf\em{Acknowledgment}}\\
The author is grateful to the anonymous referee for carefully reading this paper and providing helpful comments. The author is also grateful to the production editor for ensuring that the finished article is nicely presented.

% BibTeX
\nocite{*}     %remove nocite once all refs are cited in the text
\bibliographystyle{unsrt}
%\bibliography{masonComplexitybibdata}

\begin{thebibliography}{10}

\bibitem{Mason}
Jonathan W.~D. Mason.
\newblock Consciousness and the structuring property of typical data.
\newblock {\em Complexity}, 18(3):28--37, 2013.

\bibitem{BCM}
E~L Bienenstock, L~N Cooper, and P~W Munro.
\newblock {Theory for the development of neuron selectivity - orientation
  specificity and binocular interaction in visual-cortex}.
\newblock {\em {Journal of Neuroscience}}, {2}({1}):{32--48}, {1982}.

\bibitem{KRB}
A~Kirkwood, M~G Rioult, and M~F Bear.
\newblock {Experience-dependent modification of synaptic plasticity in visual
  cortex}.
\newblock {\em {Nature}}, {381}({6582}):{526--528}, {JUN 6} {1996}.

\bibitem{DB}
S~M Dudek and M~F Bear.
\newblock {Homosynaptic Long-Term Depression in Area CA1 of Hippocampus and
  Effects of N-Methyl-D-Aspartate Receptor Blockade}.
\newblock {\em {Proceedings of The National Academy of Sciences of The United
  States of America}}, {89}({10}):{4363--4367}, {MAY 15} {1992}.

\bibitem{Wagemans}
Johan; et~al Wagemans.
\newblock A century of gestalt psychology in visual perception: Ii. conceptual
  and theoretical foundations.
\newblock {\em Psychological Bulletin}, 138(6):1218--1252, 2012.

\bibitem{Shinkareva}
Svetlana~V. Shinkareva, Vicente~L. Malave, Robert~A. Mason, Tom~M. Mitchell,
  and Marcel~Adam Just.
\newblock {Commonality of neural representations of words and pictures}.
\newblock {\em {Neuroimage}}, {54}({3}):{2418--2425}, {FEB 1} {2011}.

\bibitem{Mansour}
T~Mansour.
\newblock {\em Combinatorics of Set Partitions, Discrete Mathematics and its
  Applications}.
\newblock CRC Press, Boca Raton, FL., 2012.

\bibitem{Quiroga}
R.~Quian Quiroga, G.~Kreiman, C.~Koch, and I.~Fried.
\newblock Sparse but not 'grandmother-cell' coding in the medial temporal
  lobe.
\newblock {\em Trends in Cognitive Sciences}, 12(3):87 -- 91, 2008.

\bibitem{Graham}
D.~J. Graham and D.~J. Field.
\newblock Natural images: coding efficiency.
\newblock In Larry~R. Squire, editor, {\em Encyclopedia of Neuroscience}, pages
  19--27. Academic Press, Oxford, 2009.

\bibitem{Tononi}
Giulio Tononi.
\newblock {Consciousness as Integrated Information: a Provisional Manifesto}.
\newblock {\em {Biological Bulletin}}, {215}({3}):{216--242}, {DEC} {2008}.

\bibitem{Oizumi}
Masafumi Oizumi, Larissa Albantakis, and Giulio Tononi.
\newblock From the phenomenology to the mechanisms of consciousness: Integrated
  information theory 3.0.
\newblock {\em PLoS Comput Biol}, 10(5):e1003588, 05 2014.

\bibitem{Balduzzi}
David Balduzzi and Giulio Tononi.
\newblock Qualia: The geometry of integrated information.
\newblock {\em PLoS Comput Biol}, 5(8):e1000462, 08 2009.

\bibitem{Balduzzi2}
David Balduzzi and Giulio Tononi.
\newblock Integrated information in discrete dynamical systems: Motivation and
  theoretical framework.
\newblock {\em PLoS Comput Biol}, 4(6):e1000091, 06 2008.

\bibitem{Barrett}
Adam~B. Barrett and Anil~K. Seth.
\newblock Practical measures of integrated information for time-series data.
\newblock {\em PLoS Comput Biol}, 7(1):e1001052, 01 2011.

\bibitem{Edelman}
G~M Edelman and G~Tononi.
\newblock {\em A Universe of Consciousness: How Matter Becomes Imagination}.
\newblock Basic Books, 2000.

\bibitem{Ascoli}
G~A Ascoli.
\newblock {The complex link between neuroanatomy and consciousness}.
\newblock {\em {Complexity}}, {6}({1}):{20--26}, {Sep} {2000}.

\bibitem{Sporns}
O~Sporns.
\newblock {Network analysis, complexity, and brain function}.
\newblock {\em {Complexity}}, {8}({1}):{56--60}, {Sep} {2002}.

\bibitem{Miyawaki}
Yoichi Miyawaki, Hajime Uchida, Okito Yamashita, Masa-aki Sato, Yusuke Morito,
  Hiroki~C. Tanabe, Norihiro Sadato, and Yukiyasu Kamitani.
\newblock {Visual Image Reconstruction from Human Brain Activity using a
  Combination of Multiscale Local Image Decoders}.
\newblock {\em {Neuron}}, {60}({5}):{915--929}, {Dec 11} {2008}.

\bibitem{Krejic}
Natasa Krejic and Natasa~Krklec Jerinkic.
\newblock {Stochastic Gradient Methods for Unconstrained Optimization}.
\newblock {\em {Pesquisa Operacional}}, 34:373 -- 393, 12 2014.

\bibitem{Nocedal}
J~Nocedal and S~Wright.
\newblock {\em Numerical Optimization}.
\newblock Springer Series in Operations Research and Financial Engineering.
  Springer-Verlag New York, 2nd edition, 2006.

\bibitem{Conn}
Scheinberg~K Conn A~R and Vicente~L N.
\newblock {\em Introduction to Derivative-Free Optimization}.
\newblock MPS-SIAM Series on Optimization. SIAM, 2009.

\end{thebibliography}
% Insert the bibliography data here.

\end{document}